\documentclass[11pt]{article}
\RequirePackage[colorlinks,citecolor=blue,urlcolor=blue,linkcolor=blue]{hyperref}
\usepackage{url}

\DeclareSymbolFont{slenderlargesymbols}{OMX}{ccex}{m}{n}
\DeclareMathSymbol{\prod}{\mathop}{slenderlargesymbols}{"51}
\usepackage{latexsym,amssymb,amsmath,amsthm,graphics,graphicx,float,psfrag, epsfig, color, authblk, enumerate}
\usepackage{mathabx}
\newtheorem{theorem}{Theorem} 
\newtheorem{lemma}[theorem]{Lemma}
\newtheorem{proposition}[theorem]{Proposition}
\newtheorem{corollary}[theorem]{Corollary}
\newtheorem{definition}[theorem]{Definition}

\newcommand{\gtilde}{\tilde{g}}
\newcommand{\ctilde}{\tilde{c}}
\newcommand{\R}{\mathbb{R}}
\newcommand{\calB}{\mathcal{B}}
\newcommand{\A}{\mathcal{A}}
\newcommand{\calAp}{\A_+}
\newcommand{\calWp}{\mathcal{W}_+}
\newcommand{\calBp}{\mathcal{B}_+}

\newcommand{\eps}{\epsilon}
\newcommand{\rhotilde}{\tilde{\rho}}
\newcommand{\xtilde}{\tilde{x}}
\newcommand{\ztilde}{\tilde{z}}
\newcommand{\PP}{\mathbb{P}}
\newcommand{\thetabar}{\overline{\theta}}
\newcommand{\thetacheck}{\widecheck{\theta}}
\newcommand{\zetacheck}{\rho}

\newcommand{\spherek}{S^{k-1}}
\newcommand{\spherel}{S^{\ell-1}}
\newcommand{\calN}{\mathcal{N}}
\newcommand{\tran}{t}
\newcommand{\red}[1]{{\color{black} #1}}

\newcommand{\xo}{x_0}

\newcommand{\Wt}{W^\tran}

\newcommand{\Wplus}[1]{W_{+, #1}}

\newcommand{\Wpx}{\Wplus{x}}
\newcommand{\Wpy}{\Wplus{y}}

\newcommand{\indwx}{1_{w_i \cdot x > 0}}
\newcommand{\indwy}{1_{w_i \cdot y > 0}}

\newcommand{\wi}{w_i}

\newcommand{\wix}{w_i \cdot x}
\newcommand{\wixtilde}{w_i \cdot \xtilde}
\newcommand{\wiytilde}{w_i \cdot \ytilde}
\newcommand{\wiy}{w_i \cdot y}
\newcommand{\wit}{w_i^\tran}
\newcommand{\Wpxt}{\Wpx^\tran}
\newcommand{\Wpyt}{\Wpy^\tran}

\newcommand{\WpxtWpx}{\Wpx^\tran \Wpx}

\newcommand{\WpxtWpy}{\Wpx^\tran \Wpy} 

\newcommand{\hmeps}{h_{-\eps}}
\newcommand{\Gmeps}{G_{-\eps}}

\newcommand{\heps}{h_{\eps}}
\newcommand{\Geps}{G_{\eps}}

\newcommand{\wiwit}{\wi\wit}
\newcommand{\xii}{\xi_i}
\newcommand{\xiit}{\xii^\tran}
\newcommand{\xiixiit}{\xi_i\xiit}
\newcommand{\matrixleq}{\preceq}
\newcommand{\matrixgeq}{\succeq}
\newcommand{\E}{\mathbb{E}}

\newcommand{\Ndelta}{\calN_\delta}

\newcommand{\Mxyhat}{M_{\hat{x} \leftrightarrow \hat{y}}}

\newcommand{\xhat}{\hat{x}}
\newcommand{\yhat}{\hat{y}}
\newcommand{\xohat}{\hat{x}_0}

\newcommand{\ytilde}{\tilde{y}}

\newcommand{\AtA}{A^\tran A}

\newcommand{\Qxy}{Q_{x,y}}

\newcommand{\Sepsxo}{S_{\eps, \xo}}

\newcommand{\hxxo}{h_{x, \xo}}
\newcommand{\htildexy}{\tilde{h}_{x,y}}

\newcommand{\XWhatWtwo}{X_{\What, W_2}}

\newcommand{\rmax}{r_{\text{max}}}
\newcommand{\bigPioned}{\prod_{i=1}^d}
\newcommand{\Pioned}{\Pi_{i=1}^d}
\newcommand{\PiWd}{\Pi_{i=d}^1 W_{i,+,x}}

\newcommand{\PiWdy}{\Pi_{i=d}^1 W_{i,+,y}}
\newcommand{\PiWdo}{\Pi_{i=d}^1 W_{i,+,\xo}}

\newcommand{\PiWn}[1]{\Pi_{i={#1}}^1 W_{i,+,x}}
\newcommand{\PiWny}[1]{\Pi_{i={#1}}^1 W_{i,+,y}}
\newcommand{\Pipithetapii}{\prod_{i=0}^{d-1} \frac{\pi - \thetabar_i}{\pi}}
\newcommand{\zetaterm}{\sum_{i=0}^{d-1} \frac{\sin \thetabar_i}{\pi} \prod_{j=i+1}^{d-1} \frac{\pi - \thetabar_j}{\pi}}

\newcommand{\Wipxn}[1]{W_{#1, +, x}}
\newcommand{\Wipyn}[1]{W_{#1, +, y}}

\newcommand{\Wnpx}[1]{W_{#1, +, x}}
\newcommand{\Wnpy}[1]{W_{#1, +, y}}

\newcommand{\Wipx}{W_{i, +, x}}
\newcommand{\Wipy}{W_{i, +, y}}
\newcommand{\Vpx}{V_{+,x}}
\newcommand{\Vpy}{V_{+,y}}
\newcommand{\Wippx}{(W_i)_{+, x}}
\newcommand{\Wippy}{(W_i)_{+, y}}
\newcommand{\Wippxt}{(W_i)^t_{+, x}}

\newcommand{\Win}[1]{W_{#1}}
\newcommand{\Wi}{\Win{i}}

\newcommand{\What}{\hat{W}}

\newcommand{\vx}{v_{x}}
\newcommand{\vbarx}{\overline{v}_{x}}
\newcommand{\vbar}{\overline{v}}
\newcommand{\vxxo}{v_{x, \xo}}
\newcommand{\vbarxxo}{\overline{v}_{x, \xo}}

\newcommand{\vxn}{v_{x_n}}

\newcommand{\hxn}{h_{x_n}}
\newcommand{\hx}{h_x}

\newcommand{\xn}{x_{n}}

\newcommand{\Ktilde}{\tilde{K}}

\DeclareMathOperator{\Span}{span}
\DeclareMathOperator{\relu}{relu}
\DeclareMathOperator{\diag}{diag}

\DeclareMathOperator{\range}{range}

\title{Global Guarantees for Enforcing Deep Generative Priors by Empirical Risk}

\author{ Paul Hand\thanks{This work was accepted for presentation at Conference on Learning Theory (COLT) 2018.  Both authors contributed equally.  }\ \thanks{Department of Mathematics and Khoury College of Computer and Information Sciences, Northeastern University, Boston, MA} \  and Vladislav Voroninski$^*$\thanks{Helm.ai, Menlo Park, CA}}

\begin{document}

\maketitle

\begin{abstract}
We examine the theoretical properties of enforcing priors provided by generative deep neural networks via empirical risk minimization. In particular we consider two models, one in which the task is to invert a generative neural network given access to its last layer and another in which the task is to invert a generative neural network given only compressive linear observations of its last layer.  We establish that in both cases, in suitable regimes of network layer sizes and a randomness assumption on the network weights, that the non-convex objective function given by empirical risk minimization does not have any spurious stationary points. That is, we establish that with high probability, at any point away from small neighborhoods around two scalar multiples of the desired solution, there is a descent direction. Hence, there are no local minima, saddle points, or other stationary points outside these neighborhoods.  These results constitute the first theoretical guarantees which establish the favorable global geometry of these non-convex optimization problems, and they bridge the gap between the empirical success of  enforcing deep generative priors and a rigorous understanding of non-linear inverse problems.\\
\end{abstract}

\section{Introduction}
Exploiting the structure of images and natural signals has proven to be a fruitful endeavor across many domains of science. For instance, the wavelet transform, discovered by Daubechies and others \cite{daubechies1992ten}, led to the observation that natural images are sparse in the wavelet basis, enabling compression algorithms such as JPEG 2000 to tame the storage and transfer of the modern deluge of image and video data.   Principles of wavelet based image compression, combined with surprising advances in convex relaxation, have also opened the door to greatly improved signal acquisition strategies, which unlocked critical applications throughout the imaging sciences.  In particular, breaking with the dogma of the Nyquist sampling theorem, which stems from worst-case analysis, Candes and Tao, and Donoho \cite{CRT2005,Donoho, donoho2009counting}, provided a theory and practice of compressed sensing (CS), which exploits the sparsity of natural signals in the wavelet basis to design acquisition strategies with drastically lower sample complexity --- that on par with the sparsity level of the signal at hand. In particular,  using the standard basis in lieu of the wavelet basis without loss of generality, they established that to recover a vector $x \in \mathbb{R}^n$ with $k < n$ non-zero entries from $m = O(k \log n)$ observations $\langle x, a_i \rangle, i = 1,2\ldots, m$, where $a_i$ are i.i.d Gaussian, it suffices to minimize $\|x\|_1$ subject to the observations with high probability. On a practical level, compressed sensing has lead to significant reduction in the sample complexity of signal acquisition of natural images, for instance speeding up MRI imaging by an order of magnitude \cite{MRM:MRM21391}. Beyond MRI, compressed sensing has impacted many if not all imaging sciences, by providing a general tool to exploit the parsimony of natural signals to improve acquisition speed, increase SNR and reduce sample complexity. More broadly, the principled use of sparsity as a prior has led  to the development of the field of matrix completion \cite{MC}, breakthroughs in phase retrieval \cite{CSV2013, CESV2011} and blind deconvolution \cite{BDC}; and is at this point routinely utilized across applied mathematics and machine learning. 

Meanwhile, the advent of practical deep learning \cite{Goodfellow-et-al-2016} has significantly improved machine understanding of image and audio data. For instance, deep learning techniques are now the state of the art across most of computer vision and have taken the field far beyond where it stood just a few years prior.  
The success of deep learning ostensibly stems from its ability to exploit the hierarchical nature of images and other natural signals without explicit hand-engineering. There are many techniques and add-on architectural choices associated with deep learning, but many of them are non-essential from a theoretical and, to a large extent, practical perspective, with simple convolutional deep nets with Rectified Linear Units (ReLUs) achieving close to the state of the art performance on many tasks \cite{Simplicity}. The class of functions represented by such deep networks is readily interpretable as hierarchical compression schemes with exponentially many linear filters, each being a linear combination of filters in earlier layers. Constructing such compression schemes by hand would be quite tedious, if not impossible, and the biggest surprise and advantage of deep learning is that simple stochastic gradient descent (SGD) allows one to efficiently traverse this class of functions subject to potentially highly non-convex learning objectives. While this latter property has been empirically established in an impressive number of applications, it has so far eluded a completely satisfactory theoretical explanation.


In essence, compressive sensing, and its numerous extensions, consist of enforcing a sparsity prior to regularize the solution of an inverse problem. Thus, improvements in the state of the art of compressed sensing can come from better reconstruction algorithms, better design of signal measurements, or more sophisticated priors.  Virtually all of the tens of thousands of research articles in the umbrella field of compressive sensing have focused on the first two directions, taking the linear sparsity model as the de-facto prior for regularization.  Those two directions are fundamentally limited in that no approach at recovering a $k$-sparse signal with respect to a basis could succeed with fewer than $k$ measurements.
 
 Meanwhile, there have been great strides in generative modeling of images in modern machine learning that go well beyond linear sparsity models. 
Such improvements in priors on natural images beyond wavelet based approaches, when properly enforced, should enable more aggressive regularization of inverse problems, leading to lower sample complexity and higher SNR than traditional compressed sensing approaches. 
 In order to understand the potential for improving upon traditional compressive sensing, broadly speaking, as a function of advances in generative modeling, it is useful to reinterpret compressive sensing from the perspective of the field of generative modeling, a popular framework in machine learning which strives to sample from the probability distribution of natural images and other signals. Note that there is a duality between generative modeling and compression. Any compression scheme implicitly defines a generative model, by its inverse, and vice versa. In particular, wavelet based compression schemes implicitly define a generative model which attempts to sample from natural images via random sparse linear combinations of wavelet basis images.  This wavelet-based generative model is clearly too loose to capture the rich hierarchical structure of natural images, making it a sufficiently expressive yet very naive prior. 

Generative modeling has a rich history in machine learning, but only recent deep neural network based approaches to generative modeling have enabled the generation of realistic synthetic images in a variety of domains, for example by training generative adversarial networks (GANs) to find a Nash equilibrium of a non-convex game \cite{GAN, hong2017generative}; by training variational auto-encoders (VAEs) \cite{kingma2013auto, rezende2014stochastic}; and by training autoregressive models like PixelCNN \cite{oord2016pixel}, which generate pixels one-at-a-time by sampling from appropriate conditional probability distributions. GANs and VAEs map a low dimensional latent code space to a higher dimensional embedding space of images or other natural signals. For instance, if we equip the latent code space with a Gaussian distribution, the goal of generative adversarial training is to produce a deep neural network generator whose push-forward distribution is the distribution of natural images or another class of natural signals. 
Impressively, in-between the original posting of this paper and the current version of the manuscript, deep generative modeling has advanced to the point of producing high-resolution synthetic, yet extremely photorealistic, images of celebrity faces \cite{karras2017progressive}.  Further, continuous motion in the latent code space of the associated deep generative models has allowed for interpolation and continuous deformation of the resulting faces, even exhibiting equivariant properties with arithmetic operations in the latent code space corresponding to semantically meaningful image variations \cite{radford2015unsupervised}.


The scope of application of deep generative modeling to regularizing inverse problems is vast. These more sophisticated priors are recently emerging in empirical applications of many fields of imaging, such as medical imaging \cite{yang2017dagan, kelly2017deep, hammernik2017learning, quan2017compressed, dar2017transfer, adler2018learned,moeskops2017adversarial,wolterink2017generative, nie2017medical, mahmood2017unsupervised, kohl2017adversarial, mahapatra2017retinal, dai2017scan, xue2017segan}, microscopy \cite{rivenson2017deep},  inpainting \cite{mao2016image, yeh2017semantic}, superresolution \cite{sonderby2016amortised, ledig2016photo,johnson2016perceptual}, compressed sensing \cite{lohit2017convolutional, liu2017high, mousavi2015deep, mousavi2017learning}, image manipulation \cite{zhu2016generative}, and many more.  See \cite{lucas2018using} for a review of deep learning for inverse problems in imaging. Importantly, approaches that regularize inverse problems using deep generative models, have empirically been shown to improve over sparsity-based approaches, advancing the state of the art in several fields. For instance, in Magnetic Resonance Imaging, deep generative networks have enabled image reconstruction that is qualitatively of higher diagnostic quality and higher SNR than traditional compressive sensing allows, and is additionally two orders of magnitude faster than sparsity-based approaches due to the utilization of GPUs in applying convolutional neural networks \cite{mardani2017deep, mardani2017recurrent}.  This development is significant because of the tremendous potential clinical applications of diagnostic-quality real-time MRI visualization.  
Deep generative models have also empirically been used directly for compression \cite{rippel2017real}. 
In the case of compressed sensing, optimization of an empirical risk objective over the latent code space has been empirically shown to recover images from 10x fewer linear compressive measurements than sparsity-based approaches \cite{Price}. 

As the quality and reach of deep generative modeling continues to increase, signal recovery in many scenarios will benefit analogously.

As with the rest of machine learning, in the field of deep generative modeling for regularizing inverse problems, or as we refer to it the field of deep compressive sensing, empirics is far ahead of the state of theoretical justification. In this paper we initiate the rigorous study of enforcing deep generative models as priors on the solutions to inverse problems, by providing a theory of compressive sensing that goes beyond linear sparsity and into the realm of applying deep neural network based generative priors. In particular we show that under suitable randomness assumptions on the weights of a neural network and successively expansive hidden layer sizes, the empirical risk objective for recovering a latent code in $\mathbb{R}^k$ from $m$ linear observations of the last layer of a generative network, where $m$ is proportional to $k$ up to log factors, has no spurious local minima or saddle points, in that there is a descent direction everywhere except possibly small neighborhoods around two scalar multiples of the desired solution.  Our descent direction analysis is constructive; based on deterministic conditions on the neural network weights and the measurements; and relies on novel concentration bounds of certain random matrices, uncovering some interesting geometric properties of the landscapes of empirical risk objective functions for random generative multilayer networks with ReLU activations. For a generative network that achieves a greater degree of compression, the proposed scheme would enable lower sample complexity and higher SNR.  If a generative model can compress a signal to a latent code dimensionality $k$ much less then the signal's sparsity level, then compressed sensing with the generative prior may significantly outperform compressed sensing with sparsity prior in terms of sample complexity.

\subsection{Related theoretical work}
Latent code space optimizations after neural network training, and the optimization over the weights of a neural network during training, may both be interpreted as inverse problems \cite{Mallat}. The tools developed in this paper, such as the novel nonasymptotic concentration results for high dimensional Gaussians followed by a ReLU, may be of independent interest, in particular being amenable for establishing global non-asymptotic analysis regarding convergence of SGD for training deep neural networks. Our work also relates to recent trends in optimization. Traditionally, rigorous understanding of inverse problems has been limited to the simpler setting in which the optimization objective is convex. 
 More recently, there has been progress in understanding non-convex optimization objectives for inverse problems, in albeit analytically simpler situations than those involving multilayer neural networks. For instance, the authors of \cite{WrightPR, Burer1} provide a global analysis of non-convex objectives for phase retrieval and community detection, respectively, ruling out adversarial geometries in these scenarios for the purposes of optimization.  Additionally, rigorous guarantees of nonconvex recovery include other results in phase retrieval \cite{wirtinger, chen2015solving}, blind deconvolution \cite{li2016rapid, ma2017implicit,  huang2017blind}, robust subspace recovery \cite{maunu2017well}, discrete joint alignment \cite{chen2016projected}, and more.   

In related work, the authors of \cite{Price} also study inverting compressive linear observations under generative priors, by proving a restricted eigenvalue condition on the range of the generative neural network. However, they only provide a guarantee that is local in nature, in showing the global optimum of empirical risk is close to the desired solution.  The work provides no guarantees about why the global minimum of the nonconvex problem can be reached.   In addition, \cite{Sanjeev} studied inverting neural networks given access to the last layer using an analytical formula that approximates the inverse mapping of a neural network. The results of \cite{Sanjeev} are in a setting where the neural net is not generative, and their procedure is at only approximate, and, since it requires observation of the last layer, it is not readily extendable to the compressive linear observation setting. Meanwhile, the optimization problem we study can yield exact recovery, which we observe empirically via gradient descent. Most importantly, in contrast to \cite{Price,Sanjeev}, we provide a global analysis of the non-convex empirical risk objective function and constructively exhibit a descent direction at every point outside a neighborhood of the desired solution and a negative scalar multiple of it. Our guarantees are non-asymptotic, and to the best of our knowledge the first of their kind. 

\red{\subsection{Notation}  \label{sec:notation}
Before we present the main result, we now introduce notation that will be used throughout this paper. Let $[n] = \{1, \ldots, n\}$. Let $e_i$ is the $i$th standard basis element
for $i \in [n]$.  Let $\relu(x) = \max(x, 0)$ apply entrywise for $x \in \R^n$.  Let $\diag(Wx>0)$ be the diagonal matrix that is 1 in the $(i,i)$th entry if $(Wx)_i > 0$, and 0 otherwise.  Let $\calB(x, r)$ be the Euclidean ball of radius $r$ centered at $x$.  Let $\Pi_{i=d}^1 W_i = W_d W_{d-1} \cdots W_1$.  
   Let $I_n$ be the $n \times n$ identity matrix.  Let $A \matrixleq B$ mean that $B-A$ is a positive semidefinite matrix.  For matrices $A$, let $\|A\|$ be the spectral norm of $A$.  Let $\spherek$ be the unit sphere in $\R^k$. 
  For any nonzero $x \in \R^n$, let $\hat{x} = x / \|x\|_2$.  
  For a set $S$, let $|S|$ denote its cardinality. 
We will write $\gamma = O( \delta)$ to mean that there exists a positive constant C such that $\gamma \leq C \delta$, when $\gamma$ is understood to be positive.  Similarly we will write $c = \Omega( \delta)$ to mean that there exists a positive constant C such that $c \geq C \delta$.   
When we say that a constant depends polynomially on $\eps^{-1}$, that means that it is at most $C \eps^{-k}$ for some positive $C$ and positive integer $k$.  
Let $\theta_0 = \angle (x, \xo)$ and $\thetabar_1  = g(\theta_0)$ where $g$ is given by \eqref{defn-g}.
For notational convenience, we will write $a = b + O_1(\eps)$ if $\| a-b\| \leq \eps$, where the norm is understood to be absolute value for scalars, the $\ell_2$ norm for vectors, and the spectral norm for matrices.  Write $g^{\circ d}$ to denote the composition of $g$ with itself $d$ times.
Let $1_S = 1$ if S and 0 otherwise.  For nonzero $v$, let $D_v f(x)$ be the (normalized) one-sided directional derivative of $f$ at $x$ in the direction of $v$: $D_v f(x) = \lim_{t \to 0^+} \frac{f(x + t v) - f(x)}{t \|v\|_2}$.}

\subsection{Main Results}

We consider the inverse problem of recovering a vector $y_0 \in \R^{n}$ from $m \ll n$ linear measurements. To resolve the inherent ambiguity from undersampling, we assume, as a prior, that the vector belongs to the range of a $d$-layer generative neural network $G: \R^k \to \R^{n}$, with $k<n$.  To recover the vector $y_0= G(\xo)$, we attempt to find the latent code $\xo \in \R^k$ corresponding to it.  
We consider a generative network modeled by $G(x) = \relu(\Win{d} \ldots \relu (\Win{2} \relu(\Win{1} \xo)) \ldots )$, where $\relu(x) = \max(x, 0)$ applies entrywise, $W_i \in \R^{n_i \times n_{i-1}}$, $n_i$ is the number of neurons in the $i$th layer, and $k=n_0 < n_1 < \cdots < n_d = n$.  
 We consider linear measurements of $G(\xo)$ given by the sampling matrix $A \in \R^{m \times n}$ and consider $k < m \ll n$.    The problem at hand is:
\begin{alignat*}{2}
	&\text{Let: } &&\xo \in \R^{k}, A\in \R^{m \times n}, \Wi \in \R^{n_i \times n_{i-1} } \text{ for } i=1\ldots d,\\
	& &&G(x) = \relu(\Win{d} \ldots \relu (\Win{2} \relu(\Win{1} \xo)) \ldots ),\\
	& &&y_0 = G(\xo),\\
	&\text{Given: } && W_1\ldots W_d, A,  \text{  and  observations  } A y_0 ,\\
	&\text{Find: } &&\xo.
\end{alignat*}
This problem can be viewed in two ways: (1) as above, given compressive measurements of a vector with the prior information that it belongs to the output of a generative neural network, find that vector;
or (2), given compressive observations of the output of a generative neural network, find the latent code corresponding to the network's output by inverting the neural network and compression simultaneously.  

As a way to solve the above problem, we consider minimizing the empirical risk objective
\begin{align}
f(x) := \frac{1}{2} \Bigl \| A G(x) - A y_0 \Bigr\|_2^2. \label{defn-f}
\end{align}
As this objective is nonconvex, there is no \textit{a priori} guarantee of efficiently finding the global minimum \cite{Murty1987}.   Approaches such as gradient descent could in principle get stuck in local minima, instead of finding the desired global minimizer $\xo$.  

In this paper, we consider a fully-connected generative network $G:\R^k \to \R^n$ with Gaussian weights and no bias term, along with a Gaussian sampling matrix $A \in \R^{m \times n}$.  We show that under appropriate conditions and with high probability, $f$ has a strict descent direction everywhere outside two small neighborhoods of $\xo$ and a negative multiple of $\xo$.  We assume that the network is sufficiently \textit{expansive} at each layer, $n_i = \Omega(n_{i-1} \log n_{i-1})$, and that there are a sufficient number of measurements, $m = \Omega(k d \log (n_1 \cdots n_d))$.  \red{Let $D_v f(x)$ be the (normalized) one-sided directional derivative of $f$ at $x$ in the direction of $v$: $D_v f(x) = \lim_{t \to 0^+} \frac{f(x + t v) - f(x)}{t \|v\|_2}$}. \red{Let $\calB(x,r)$ be the Euclidean ball of radius $r$ centered at $x$. } Our main result is as follows:
  
\begin{theorem} \label{thm-multi-layer-XX}
Fix $\eps >0$ such that $K_1 d^8 \eps^{1/4} \leq 1$, and let $d \geq 2$.  Assume $n_i \geq c n_{i-1} \log n_{i-1}$ for all $i = 1 \ldots d$ and $m > c d  k \log \Pi_{i=1}^d n_i$.  Assume that for each $i$, the entries of $W_i$ are i.i.d. $\mathcal{N}(0, 1/n_i)$, and the entries of $A$ are i.i.d. $\mathcal{N}(0, 1/m)$ and independent from $\{W_i\}$.
Then, on an event of probability at least $1 -  \sum_{i=1}^d \ctilde n_i e^{-\gamma n_{i-1} }- \ctilde e^{-\gamma m}$, we have the following. For all nonzero $x$ and $\xo$, there exists $\vxxo \in \R^k$ such that the one-sided directional derivatives of $f$ satisfy
\begin{alignat*}{2}
&D_{-\vxxo} f(x) <0, \quad & & \forall x  \not\in  \calB(\xo, K_2 d^{3} \eps^{1/4} \|\xo\|_2) \cup \calB(-\zetacheck_d \xo, K_2 d^{13} \eps^{1/4} \|\xo\|_2) \cup \{0\},\\
&D_v f(0) < 0, & &\forall v \neq 0,
\end{alignat*}
where $\rho_d$ is a positive number that converges to $1$ as $d \to \infty$.
  Here, $c$ and $\gamma^{-1}$ are constants that depend polynomially on $\eps^{-1}$, and $\ctilde, K_1, K_2$ are universal constants.  

\end{theorem}
\red{This theorem states that for a network of fixed depth $d$, with high probability there is always a descent direction outside of  two specified, sufficiently small neighborhoods, provided that the network is Gaussian and sufficiently expansive.  Further, for such networks, zero is a local maximizer.  We note that the linear dependence of sample complexity with respect to $k$, for fixed $d$, is optimal.  The fixed $d$ regime is realistic to applications because many deep learning networks in the wild have $d$ on the order of only 10.  We also note that the theorem's scalings with respect to $\eps$, $d$, and $n_i$ are all polynomial, and not exponential, though the dependence on each of these variables could likely be improved.  While the sample complexity scaling appears to get worse for larger $d$, we note that larger $d$ allows for the possibility of generative models with lower values of $k$.  This is because the number of piecewise linear pieces in $G$ grows exponentially in $d$.}  Also, note that while the weights of any layer of the network are assumed to be i.i.d. Gaussian, there is no assumption on the independence between $W_i$ and $W_j$ for $i \neq j$.

\red{The descent direction $\vxxo$ is given by the gradient of $f$:
\begin{align*}
\vxxo &= \begin{cases}\nabla f(x) & G \text{ is differentiable at x,}\\ \lim_{\delta \to 0^+} \nabla f(x + \delta w) & \text{otherwise,} \end{cases} 
\end{align*}
where $w$ can be arbitrarily chosen such that $G$ is differentiable at $x+\delta w$ for sufficiently small $\delta$.  Such a $w$ exists by the piecewise linearity of $G$, and could be generated randomly with probability 1.   An explicit formula for $\nabla f(x)$, where it exists, is given by \eqref{defn-nabla-f} in Section \ref{sec:proofs}.  This expression for $\vxxo$ is in a form that can be computed for any $x$, even for points of nondifferentiability, as part of a gradient based algorithm.  }

This theorem will be proven by showing the sufficiency of two deterministic conditions on $G$ and $A$, and then by showing that Gaussian $G$ and $A$ of appropriate sizes satisfy these conditions with the appropriate probability.
The first deterministic condition is on the spatial arrangement of the network weights within each layer.  
\red{
\begin{definition}
We say that the matrix $W \in \R^{n \times k}$ satisfies the \emph{Weight Distribution Condition} with constant $\eps$ if for all nonzero $x,y \in \R^k$, 
\begin{align}
\Bigl \| \sum_{i=1}^n \indwx \indwy \cdot w_i w_i^t  - \Qxy \Bigr \| \leq \eps,  \text{ with } \Qxy = \frac{\pi - \theta_0}{2 \pi} I_k + \frac{\sin \theta_0}{2\pi}  \Mxyhat, \label{WDC}
\end{align}
where $w_i \in \R^k$ is the $i$th row of $W$; $\Mxyhat \in \R^{k \times k}$ is the matrix\footnote{A formula for $\Mxyhat$ is as follows.  If $\theta_0 = \angle(\xhat, \yhat) \in (0, \pi)$ and $R$ is a rotation matrix such that $\xhat$ and $\yhat$ map to $e_1$ and $\cos \theta_0 \cdot e_1 + \sin \theta_0 \cdot e_2$ respectively, then $\Mxyhat = R^t \begin{pmatrix} \cos \theta_0 & \sin \theta_0 & 0 \\ \sin \theta_0 & - \cos \theta_0 & 0 \\ 0 & 0 & 0_{k-2} \end{pmatrix} R$, where $0_{k-2}$ is a $k-2 \times k-2$ matrix of zeros.  If $\theta_0 = 0$ or $\pi$, then $\Mxyhat = \xhat \xhat^t$ or $- \xhat \xhat^t$, respectively.} such that $\xhat \mapsto \yhat$, $\yhat \mapsto \xhat$, and $z \mapsto 0$ for all $z \in \Span(\{x,y\})^\perp$;  $\xhat = x/\|x\|_2$  and $\yhat = y /\|y\|_2$;  $\theta_0 = \angle(x, y)$; and $1_S$ is the indicator function on $S$. 
\end{definition}
}

 The norm on the left hand side of \eqref{WDC} is the spectral norm.  Note that an elementary calculation\footnote{To do this calculation, take $x=e_1$ and $y = \cos \theta_0\cdot e_1 + \sin \theta_0 \cdot e_2$ without loss of generality.  Then each entry of the matrix can be determined analytically by an integral that factors in polar coordinates.} gives that $\Qxy = \E[\sum_{i=1}^n \indwx \indwy \cdot w_i w_i^t ]$ for $w_i \sim \mathcal{N}(0, I_k/n)$.  As the rows $w_i$ correspond to the neural network weights of the $i$th neuron in a layer given by $W$, the WDC provides a deterministic property under which the set of neuron weights within the layer given by $W$ are distributed approximately like a Gaussian.  The WDC could also be interpreted as a deterministic property under which the  neuron weights are distributed approximately like a uniform random variable on a sphere of a particular radius.  Note that if $x=y$, $\Qxy$ is an isometry up to a factor of $1/2$.

The second deterministic condition is that the compression matrix acts like an isometry on pairs of differences of vectors in the range of $G: \R^k \to \R^n$. 
\red{
\begin{definition}
We say that the compression matrix $A \in \R^{m \times n}$ satisfies the \emph{Range Restricted Isometry Condition (RRIC)} with respect to $G$ with constant $\eps$ if for all $x_1, x_2, x_3, x_4 \in \R^k$,
\begin{align}
\Bigl| \Bigl \langle A \bigl( G(x_1) - G(x_2) \bigr), A \bigl( G(x_3) - G(x_4) \bigr) \Bigr\rangle  &- \Bigl \langle  G(x_1) - G(x_2) ,  G(x_3) - G(x_4)  \Bigr \rangle \Bigr| \notag \\&\leq \eps \| G(x_1) - G(x_2) \|_2  \| G(x_3) - G(x_4) \|_2.
\end{align}
\end{definition}}
We can now state our main deterministic result.  

\begin{theorem} \label{thm-multi-layer-deterministic}
Fix $\eps >0$ such that $K_1 d^8 \eps^{1/4} \leq 1$, and let $d \geq 2$.  Suppose that $G$ is such that $W_i$ has the  WDC with constant $\eps$ for all $i = 1 \ldots d$.  Suppose $A$ satisfies the RRIC with respect to $G$ with constant $\eps$.  Then, for all nonzero $x$ and $\xo$, there exists $\vxxo \in \R^k$ such that the one-sided directional derivatives of $f$ satisfy
\begin{alignat*}{2}
&D_{-\vxxo} f(x) < \red{-K_3 \frac{\sqrt{\eps} d^3}{2^d} \max(\|x\|_2, \|\xo\|_2)}, \quad & & \forall x  \not\in  \calB(\xo, K_2 d^{3} \eps^{1/4} \|\xo\|_2) \cup \calB(-\zetacheck_d \xo, K_2 d^{13} \eps^{1/4} \|\xo\|_2) \cup \{0\},\\
&D_y f(0) < \red{- \frac{1}{8 \pi 2^d} \|\xo\|_2}, & &\forall y \neq 0,
\end{alignat*}
where $\rho_d$ is a positive number that converges to $1$ as $d \to \infty$,  and $K_1$, $K_2$, and $K_3$ are universal constants.
\end{theorem}
\red{Note that the $2^d$ scaling in the bounds is an artifact of the scaling of the problem and does not indicate a vanishingly small derivative.  Roughly speaking, the $\relu$ activation functions zero out roughly half of its arguments.  Hence, while $W_i$ has spectral norm approximately $1$, the rows of $W_i$ that are retained by the $\relu$ will have spectral norm approximately $1/2$.  Thus, $f(x)$ itself is on the order of $2^{-d}$ under the RRIC and WDC for appropriately small $\eps$.  }

In the case that $A = I_n$, the RRIC is trivially satisfied, and we get the following corollary about inverting multilayer neural networks.
\begin{corollary}[Approximate Invertibility of Multilayer Neural Networks]
If $G$ is a $d$-layer neural network such that $W_i$ satisfies the WDC with constant $\eps$ for all $i = 1 \ldots d$, then the function $f(x) = \|G(x) - G(\xo)\|_2$ has no stationary points outside of a neighborhood around $\xo$ and $-\rho_d \xo$. 
\end{corollary}

In the case of a Gaussian network with Gaussian measurements, the WDC and RRIC are satisfied with high probability if the network is sufficiently expansive and there are a sufficient number of measurements. 
\begin{proposition} \label{thm-multi-layer}
Fix $0 < \eps < 1$.  Assume $n_i \geq c n_{i-1} \log n_{i-1}$ for all $i = 1 \ldots d$ and $m > c d  k \log \Pi_{i=1}^d n_i$.  Assume the entires of $W_i$ are i.i.d. $\mathcal{N}(0, 1/n_i)$, and the entries of $A$ are i.i.d. $\mathcal{N}(0, 1/m)$.  
Then, $W_i$ satisfies the WDC with constant $\eps$ for all $i$ and $A$ satisfies the RRIC with respect to $G$  with constant $\eps$ with probability at least $1 -  \sum_{i=1}^d \ctilde n_i e^{-\gamma n_{i-1} }- \ctilde e^{-\gamma m}$.  Here, $c$ and $\gamma^{-1}$ are constants that depend polynomially on $\eps^{-1}$, and $\ctilde$ is a universal constant.
\end{proposition}
As stated after Theorem \ref{thm-multi-layer-XX}, no assumption is made on the independence between $W_i$ and $W_j$ for $i \neq j$.
While Proposition \ref{thm-multi-layer} is stated for $A \in \R^{m \times n}$ with i.i.d. Gaussian entries, it also applies in the case of any random matrix that satisfies the following concentration of measure condition:
\begin{align*}
\PP \bigl(| \|Ax\|_2^2  -  \|x\|_2^{2} | \geq \epsilon \|x\|_2^2 \bigr) \leq 2 e^{-m c_0(\epsilon)},
\end{align*}
for any fixed $x \in \R^n$, where $c_0(\epsilon)$ is a positive constant depending only on $\epsilon$.   In particular, Proposition \ref{thm-multi-layer} and hence Theorem \ref{thm-multi-layer-XX} extends to the case of where the entries of $A$ are independent Bernoulli random variables (and the entries of $W_i$ are Gaussian).  See \cite{Baraniuk2008} for more.

\subsection{Discussion}

In this paper, we provide the first rigorous global analysis of the efficacy of enforcing generative neural network priors.  
We show that if a generative neural network has Gaussian weights and is sufficiently expansive at each layer, then, with high probability, the empirical risk objective applied to the network output has no spurious local minima or saddle points outside two small neighborhoods around the global optimum and a negative reflection of it.  Further, if the output of the network is subject to random Gaussian compressive measurements, then the same conclusion holds with information theoretically optimal sample complexity with respect to the latent code dimensionality.  That is, a convergent gradient descent scheme will approximately invert the generative network, even in the presence of a sufficient number of compressive measurements.  
  


As this theoretical work is the first of its kind, it leaves open many important questions deserving further research.  Because any particular generative network is unlikely to contain an observed image \textit{exactly}, it is important to establish a similar guarantee to Theorem \ref{thm-multi-layer-XX} in the case that the observed image is not in the range of the generative network.  This line of work includes establishing noise tolerance and robustness to outliers, both of which have been established for sparsity-based compressed sensing. Such results would provide even further theoretical support for several empirical observations about enforcing generative priors via an optimization over latent code space \cite{Price, zhu2016generative}, including empirical robustness of inverting generative models \cite{lipton2017precise}.  In particular, it could help explain the significant observation that generative priors can mitigate against adversarial examples \cite{ilyas2017robust, song2017pixeldefend}, which are minor and sometimes imperceptible modifications to images that lead to catastrophic misclassification by neural networks \cite{szegedy2013intriguing}.  Robustness against adversarial examples is important for the security of machine learning systems \cite{yuan2017adversarial, akhtar2018threat}, in particular those that will be part of self-driving cars.  


In this work, we assume that weights of the generative network are modeled by Gaussians.  There is empirical evidence justifying this assumption for trained neural networks \cite{Sanjeev}.  Additionally, previous theoretical work with neural networks, in the area of classification, has also studied Gaussian networks \cite{giryes2016deep}.  As was the case with compressed sensing, where theoretical developments with Gaussians inspired subsequent theory with more realistic measurement models, the novel results of this paper motivate additional theoretical  with more complicated assumptions on the weights of generative networks.  Further, while the generative model assumed in this paper captures key structural elements of real neural networks (each layer acting as a nonlinear function of a linear transformation), this work also motivates establishing similar results for more complicated network structures, including bias terms, convolutional layers, max-pooling, and more.  

Most importantly, we provide in this paper a theoretical framework for studying the enforcement of deep generative priors via empirical risk as a means of regularization on inverse problems. Besides compressive sensing with linear measurements, there are a myriad of inverse problems that may benefit from such an approach. One particularly exciting example is the field of phase retrieval, which is critical in the biological sciences for X-ray crystallography and modern techniques like XFEL-imaging, which is a promising approach that may lead to breakthroughs in understanding of proteins and other molecular structures. Phase retrieval involves recovering vectors from quadratic observations, and enforcing linear sparsity priors subject to such quadratic measurements has been met with potentially fundamental limitations of polynomial time algorithms. In particular, while the theoretically optimal sample complexity of sparse phase retrieval is $O(\red{s} \log n)$, \red{where $s$ is the sparsity of the signal}, it is potentially unobtainable via polynomial time algorithms \cite{barak2016nearly}, which have so far only produced $O(\red{s}^2 \log n)$ efficient reconstruction schemes \cite{li2013sparse}. This bottleneck in sample complexity makes improvements in signal priors critical for the field of phase retrieval to advance.  \red{The present work indicates that it may be possible to use generative priors for problems such as phase retrieval.  If a bound similar to that in our Theorem is true in phase retrieval, it would mean that recovery could be possible with $O(k)$ measurements, where $k$ is the latent code dimensionality.  This could beat sparsity based approaches both because the scaling is linear in the signal's latent dimensionality of the representation, and because $k$ can be smaller than $s$ for the very same signal.}

 More significantly, the regime of using generative modeling as a means of regularization opens new doors for improving the workflow of biological scientists. In modern phase retrieval, modeling assumptions which aid in lowering sample complexity and increasing SNR are all hand-coded, making the process extremely tedious. In contrast, deep generative modeling simply requires obtaining a dataset of previously reconstructed molecular structures, which are easily available in extensive databases amassed over the years of practice of crystallography \cite{velankar2009pdbe}. One may envision training generative models on such datasets and using the resulting neural network priors to regularize the inverse problem of phase retrieval, tabula rasa, and potentially more effectively than hand-modeling ever could, as has been witnessed in the field of computer vision. This makes possible recovering the structure of biological molecules without explicit modeling, freeing up scientists to focus on innovating on new imaging modalities instead of grappling with the tedium of hand-coding their prior knowledge to solve the resulting inverse problems. More broadly, combining the power of deep generative modeling with modern methods of optimization and signal recovery, allows potentially paradigm shifting improvements to the empirical sciences, by taking a data-driven artificial intelligence approach to signal recovery.

\section{Proofs} \label{sec:proofs}

The theorems are proven by a concentration argument.  We show that $\vxxo \in \R^k$ concentrates around a particular $\hxxo \in \R^k$ that is a continuous function of nonzero $x,\xo$ and is zero only at $x = \xo$ and $x = - \rho_d \xo$.  Before we sketch the proof below, we introduce some useful quantities.

In order to analyze which rows of a matrix $W$ are active when computing $\relu(Wx)$, we let 
\[
\Wpx = \diag(Wx>0) W.
\]
For a fixed $W$, the matrix $\Wpx$ zeros out the rows of $W$ that do not have a positive dot product with $x$.  Alternatively put, $\Wpx$ contains weights from only the neurons that are active for the input $x$.  We also define $\Wnpx{1} = (W_1)_{+,x} =  \diag(W_{1}x>0) W_1$ and  
\begin{align*}
\Wipx = \diag(\Wi \Wnpx{i-1} \cdots \Wnpx{2} \Wnpx{1} x > 0) \Wi.
\end{align*}
The matrix $\Wipx$ consists only of the neurons in the $i$th layer that are active if the input to the first layer is  $x$.  Additionally, it will be useful to control how the operator $x\mapsto W_{+,x} x$ distorts angles. In order to study this, we define 
\begin{align}
g(\theta) := \cos^{-1} \Bigl( \frac{ (\pi - \theta) \cos \theta + \sin \theta}{\pi} \Bigr). \label{defn-g}
\end{align}

We now specify the choice of $\vxxo$ as follows.
At any $x \in \R^k$ such that $G$ is differentiable at $x$,
\begin{align}
 \nabla f(x) = (\PiWd)^t \AtA (\PiWd) x -  (\PiWd)^t \AtA (\PiWdo) \xo. \label{defn-nabla-f}
\end{align}
Let $w\in \R^k$ be such that $G$ is differentiable at $x+\delta w$ for sufficiently small $\delta$.  Such a $w$ exists by the piecewise linearity of $G$. 
 Let
\begin{align}
\vxxo &= \begin{cases}\nabla f(x) & G \text{ is differentiable at x,}\\ \lim_{\delta \to 0^+} \nabla f(x + \delta w) & \text{otherwise.} \end{cases} \label{defn-vxxo}
\end{align}
Note that the first part of the definition of $\vxxo$ can be viewed as the special case of the second part of the definition with $w=0$.  When $G$ is not differentiable at $x$, multiple values of $\vxxo$ are consistent with the above definition.  These values correspond to the multiple choices of $w$.  The concentration analysis applies simultaneously for all appropriate $w$ because of uniformity in the concentration results below.  

A sketch of the proof is as follows:  
\begin{itemize}
\item The WDC and RRIC imply that
\begin{align*}
\vxxo &\approx (\PiWd)^t  (\PiWd) x -  (\PiWd)^t (\PiWdo) \xo =: 
\vbarxxo,
\end{align*}
uniformly over nonzero $x$ and $\xo$.  See the proof of Theorem \ref{thm-multi-layer-deterministic} in Section \ref{sec:deterministic-theorem}.
\item The WDC implies that
\begin{align*}
\vbarxxo &\approx  - \frac{1}{2^d} \Bigl( \prod_{i=0}^{d-1} \frac{\pi - \thetabar_i}{\pi}  \Bigr)\xo 
+ \frac{1}{2^d} \left[ x - \sum_{i=0}^{d-1} \frac{\sin \thetabar_i}{\pi}  \Bigl( \prod_{j=i+1}^{d-1} \frac{\pi - \thetabar_j}{\pi}  \Bigr)  \frac{\|\xo\|_2}{\|x\|_2} x  \right] =: \hxxo, 
\end{align*}
uniformly over nonzero $x$ and $\xo$, where
$\thetabar_i = g(\thetabar_{i-1})$, $\thetabar_0 = \angle(x,\xo)$, and $\hxxo$ is continuous for nonzero $x, \xo$.  See Sections \ref{sec:angle-contraction} and  \ref{sec:concentration-no-compression}.  
\item Direct analysis shows that $\hxxo \approx 0$ only within a neighborhood of $\xo$ and $-\rho_d \xo$.   See Section \ref{sec:zeros-hxxo}.
\item Arguments from probabilistic concentration theory establish that the WDC and RRIC with high probability for Gaussian matrices of appropriate dimensions.  See Sections \ref{sec:WDC-Gaussian} and \ref{sec:RRIC-Gaussian}, respectively.  These together establish Proposition \ref{thm-multi-layer}.
\end{itemize}
Theorem \ref{thm-multi-layer-XX} is the combination of Theorem \ref{thm-multi-layer-deterministic} and Proposition \ref{thm-multi-layer}.

\red{
The proof capitalizes on the structure of the $\relu$ nonlinearities because it considers points in the range of the generator $G$ to lie in the union of finitely many subspaces, each given by the range of all possible matrices $\Wipx$.   Probabilistic concentration of these matrices requires a bound on the maximum number of such subspaces.  These bounds would be worse or possibly infinite for nonlinearities other than $\relu$.  While the nondeterministic result is stated for Gaussian weight matrices $W_i$, the same analysis would extend to weight matrices whose rows are given by a uniform distribution over a sphere of appropriate radius, as the proof capitalizes on rotational invariance of the neuronal weights.
}


\subsection{Approximate angle contraction property $\Wpx$} \label{sec:angle-contraction}
In the concentration result of the next section, we will make use  of the fact that the angle between $\Wpx x$ and $\Wpy y$ is approximately $g\bigl(\angle(x, y) \bigr)$ if the WDC holds.  As Figure \ref{fig:plot-of-g} shows, $g$ is monotonic and less than the identity.  Thus, the mapping $x \mapsto W_{+, x} x$ has an approximate angle contraction property in the sense of the following lemma.

\begin{figure}[h!]
\label{fig:plot-of-g}
\begin{center}
\includegraphics[width=0.9 \linewidth]{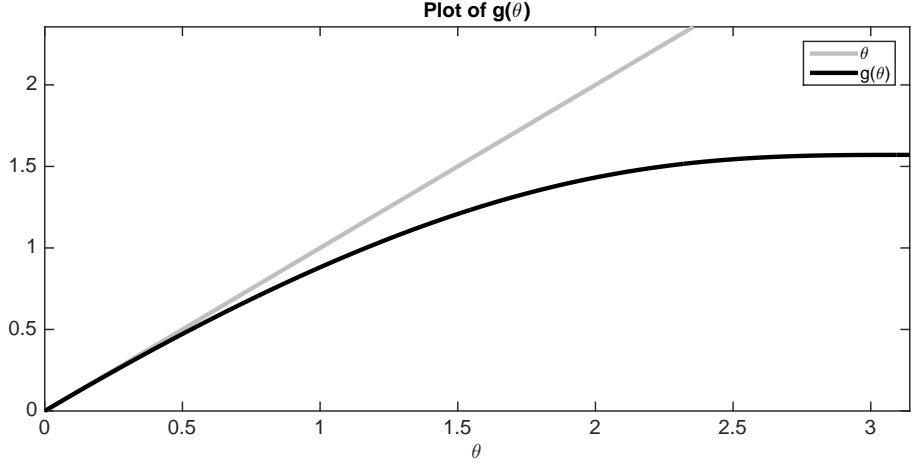}
\end{center}
\caption{A plot of $g(\theta)$ from equation \eqref{defn-g}.}
\end{figure}


\begin{lemma} \label{lemma:control-thetatwo} 
Fix $0 < \eps < 0.1$.  Let $W \in \R^{n \times k}$ satisfy the WDC with constant $\eps$.  We have
$\theta_1 := \angle( \Wpx x,  \Wpy y)$ is well-defined for all $x\neq0, y \neq0$,  and 
\[
| \theta_1 - g(\theta_0) | \leq 4 \sqrt{\eps},
\]
where $\theta_0 = \angle(x,y)$ and $g$ is defined by \eqref{defn-g}.
\end{lemma}
\begin{proof}
It suffices to establish that for all $x \neq 0, y \neq 0$,
\[
\Biggl |\cos \theta_1 -  \frac{(\pi - \theta_0) \cos \theta_0 + \sin \theta_0}{\pi} \Biggr | \leq 5 \eps.
\]
Without loss of generality, consider only $x, y \in \spherek$. 
By the WDC, $\|\Wpxt \Wpy - \Qxy \| \leq  \eps$ for all $x, y \in \spherek$.   Let 
\begin{align*}
\delta_1 &= \langle x, (\Wpxt \Wpy -  \Qxy) y \rangle\\
 \delta_2 &= \langle x, (\Wpxt \Wpx -  I_{k}/2) x\rangle\\
  \delta_3 &= \langle y, (\Wpyt \Wpy -  I_{k}/2) y \rangle.
\end{align*}
We have $\max(|\delta_1|, |\delta_2|, |\delta_3|) \leq \eps$ for all $x,y \in S^{k-1}$.   
As $|\delta_2|\leq \eps$, $|\delta_3| \leq \red{\eps}$, and $\eps < 1/2$, 
$ \Wpx x \neq 0$ and $\Wpy y \neq 0$.  Thus, $\theta_1$ is well-defined.  We have
\begin{align*}
\cos \theta_1  &= \frac{\langle \Wpx x, \Wpy y \rangle}{\| \Wpx x\|_2 \| \Wpy y\|_2} \\
&= \frac{\langle  x, \Wpxt \Wpy y \rangle}{\sqrt{ \langle x, \WpxtWpx x \rangle \langle y, \WpxtWpx y \rangle}}\\
&= \frac{\langle x, \Qxy y \rangle + \delta_1 }{\frac{1}{2} \sqrt{(1+2 \delta_2) (1 + 2 \delta_3)}}
\end{align*}
Thus, 
\begin{align*}
| \cos \theta_1 - 2 \langle x, \Qxy y \rangle | &\leq  2 |\langle x, \Qxy y\rangle| \Biggl |1 - \frac{1}{\sqrt{(1 + 2\delta_2) (1 + 2\delta_3)}} \Biggr| + 2 \delta_1 \frac{1}{\sqrt{(1+ 2\delta_2) (1 + 2\delta_3)}}  \\
& \leq  \left |1 - \frac{1}{(1 - 2 \eps)} \right | + 2 \eps \frac{1}{(1- 2 \eps)}  \\
& \leq 5\eps
\end{align*}
where the second line follows as $2 |\langle x, \Qxy y \rangle | \leq 2 \|\Qxy\| \leq 1$ and $\max(|\delta_1|, |\delta_2|, |\delta_3|) \leq \eps$, and the last line follows because $\eps < 0.1$.   The proof is concluded by noting that $2 \langle x, \Qxy y \rangle = \frac{1}{\pi} \bigl[ (\pi - \theta_0) \cos \theta_0 + \sin \theta_0 \bigr]$. 

\end{proof}

\subsection{Concentration of terms without compression} \label{sec:concentration-no-compression}

At points of differentiability $x$, we  have 
\begin{align*}
\vxxo =(\PiWd)^t \AtA (\PiWd) x -  (\PiWd)^t \AtA (\PiWdo) \xo.
\end{align*} 
In this section, we prove that the WDC establishes concentration for the terms in $\vxxo$ uniformly in $x$ and $\xo$ in the compressionless case of $A = I_n$.  The case with compression will use the concentration of these terms too.

\begin{lemma}\label{lemma:Gradient-multilayer-concentration} 
Fix $0 < \eps <  d^{-4}/(16 \pi)^2$ and $d \geq 2$.    
Suppose that $W_i \in \R^{n_i \times n_{i-1}}$ satisfies the WDC with constant $\eps$ for $i = 1\ldots d$.
Define
\[
\htildexy =  \frac{1}{2^d} \left[ \Bigl( \prod_{i=0}^{d-1} \frac{\pi - \thetabar_i}{\pi}  \Bigr)y 
+ \sum_{i=0}^{d-1} \frac{\sin \thetabar_i}{\pi}  \Bigl( \prod_{j=i+1}^{d-1} \frac{\pi - \thetabar_j}{\pi}  \Bigr)  \frac{\|y\|_2}{\|x\|_2} x  \right],
\] 
where $\bar{\theta}_i = g(\thetabar_{i-1})$ for $g$ given by \eqref{defn-g} and $\thetabar_0 = \angle(x,y)$.  For all $x \neq 0$ and $y \neq 0$, 
\begin{align}
\|( \PiWd)^t ( \PiWdy) y - \htildexy  \|_2 &\leq 24 \frac{d^3 \sqrt{\eps}}{2^{d}} \|y\|_2, \text{ and } \label{vxyhxy-d-cubed}  \\
\bigl \langle  (\PiWd) x, ( \PiWdy) y \bigr \rangle  &\geq \frac{1}{4 \pi} \frac{1}{2^d} \|x\|_2 \|y\|_2. \label{piwxpiwyy-lowerbound}
\end{align}
\end{lemma}
\begin{proof} 
\ 

\textit{Part I: Assembling some useful bounds.}  

Define $x_0 = x$, $y_0=y$, 
\[
x_d := \left( \Pi_{i=d}^1 W_{i,+,x} \right) x = (W_{d,+,x}W_{d-1,+,x}\ldots W_{1,+,x}) x
= W_{d,+,x} x_{d-1} = (W_d)_{+,x_{d-1}} x_{d-1},
\]
and analogously $y_d = \left( \Pi_{i=d}^1 W_{i,+,y} \right) y$, where $\Wippx = \diag(W_i x > 0) W_i$.

By the WDC, we have for all $i = 1 \ldots d$, for all $x\neq 0, y \neq 0$,
\begin{align}
\| \Wippxt \Wippy - Q_{x,y}\| \leq \eps, \label{defn-E1}
\end{align}
  In particular, $\|\Wipx^t \Wipy - Q_{x_{i-1}, y_{i-1}} \| \leq \eps$.   
We now detail several bounds that follow from the WDC.  Most immediately, we have that for all $x \neq 0$ and for all $i = 1 \ldots d$, 
\begin{align}
\Bigl \| \Wipx^t \Wipx - \frac{1}{2} I_{n_{i-1}} \Bigr\| \leq \eps, \label{defn-E1-symmetric-case}
\end{align}
and consequently, 
\[
\frac{1}{2} - \eps \leq \| \Wipx\|^2 \leq \frac{1}{2} + \eps.
\]
Hence,
\begin{align}
\|\PiWd\| \| \PiWdy \|  \leq \frac{1}{2^d} (1 + 2 \eps)^d = \frac{1}{2^d} e^{d \log(1+2 \eps)} \leq \frac{1 + 4 \eps d}{2^d}, \label{bound-piwixiy-on-E1}
\end{align}
where we used that $\log(1+z) \leq z$, $e^z \leq 1 + 2 z$ for $z < 1$, and $2 d \eps < 1$. 
We also have for all $x\neq 0$ that 
\begin{align}
\sqrt{\frac{1}{2} - \eps} \| x_{i-1} \|_2 \leq \| x_i \|_2 \leq \sqrt{\frac{1}{2} + \eps}\|x_{i-1}\|_2.  \label{bound_xi_ximinusone}
\end{align}
Thus, we have that for all $x,y\neq 0$, 
\begin{align*}
\left(\frac{1-2\epsilon}{1+2\epsilon}\right)^{d/2} \frac{\|y\|_2}{\|x\|_2} \leq \frac{\|y_d\|_2}{\|x_d\|_2} \leq \left(\frac{1+2\epsilon}{1-2\epsilon}\right)^{d/2} \frac{\|y\|_2}{\|x\|_2}. 
\end{align*}
Note that if $4 \eps d < 1$, 
\[
\Bigl(\frac{1+2 \eps}{1 - 2 \eps} \Bigr)^{d/2} \leq (1 + 8 \eps)^{d/2} =  e^{\frac{d}{2} \log(1+8\eps)} \leq e^{4 d \eps} \leq 1 + 8 d \eps
\]
where we used that $ \frac{1+z}{1-z} \leq 1+4z$ for $0\leq z\leq \frac{1}{2}$, $e^z \leq 1+2z$ for $0\leq z\leq1$, and $4 d \eps < 1$.  As $1-z \leq (1+z)^{-1}$ for $z>0$, we also have
\[
\Bigl(\frac{1-2 \eps}{1 + 2 \eps} \Bigr)^{d/2} \geq 1 - 8 d \eps,
\]
and thus, if $4 \eps d<1$,
\begin{align}
(1-8 d \eps) \frac{\|y\|_2}{\|x\|_2} \leq \frac{\|y_d\|_2}{\|x_d\|_2} \leq (1+8 d \eps) \frac{\|y\|_2}{\|x\|_2}. \label{bound-yd-over-xd-on-E1}
\end{align}

By Lemma \ref{lemma:control-thetatwo}, the WDC implies that $\theta_{\Vpx x, \Vpy y}$ is well-defined for all nonzero $x$ and $y$, and for $V = W_1, \ldots, W_d$, and
\begin{align}
| \theta_{\Vpx x, \Vpy y} - g(\theta_{x,y})| \leq \delta \quad \text{ for all } x\neq0, y\neq 0, V = W_1, \ldots, W_d, \label{theta-vplusx-vplusy}
\end{align}
where $g(\theta)$ is given by \eqref{defn-g},  and $\delta := 4 \sqrt{\eps}$.  
  Define $\theta_{i} :=\angle (x_{i}, y_{i}) \in [0, \pi]$.  
Note that by \eqref{theta-vplusx-vplusy}, we have\footnote{See Section \ref{sec:notation} for the meaning of $O_1$.}  $\theta_d = g(\theta_{d-1}) + O_1(\delta)$ for all $d$, and thus
$\theta_d = g(g(\cdots g(g(\theta_0) + O_1(\delta) ) +O_1(\delta)  \cdots) + O_1(\delta))   + O_1(\delta)$. Because $|g'(\theta)| \leq 1$ for all $\theta$ and because $\thetabar_d = g(g(\cdots g(\theta_0)\cdots)) = g^{\circ d} (\theta_0)$, we have
\begin{align}
|\theta_d - \overline{\theta}_d | \leq d \delta = 4 d \sqrt{\eps}. \label{bound-thetad-thetadbar}
\end{align}

\textit{Part II: Establishing \eqref{piwxpiwyy-lowerbound}}\\
We have that $\cos \theta_d \geq 3/(4 \pi)$ by combining \eqref{bound-thetad-thetadbar},  $\thetabar_d \leq \cos^{-1}(\frac{1}{\pi})$ for $d \geq 2$,  and $4 \pi d 4 \sqrt{\eps} \leq 1$.  Additionally, by \eqref{bound_xi_ximinusone}, we have $\|x_d\|_2 \|y_d\|_2 \geq \|x\|_2 \|y\|_2 \bigl( \frac{1}{2} - \eps  \bigr)^d \geq  \|x\|_2 \|y\|_2 \frac{1 - 2 d \eps}{2^d}$, where the last inequality holds as $\eps \leq 1/2$. If $2 d \eps \leq 2/3$, we have 
\[
\bigl \langle  (\PiWd) x, ( \PiWdy) y \bigr \rangle  = \cos(\theta_d) \|x_d\|_2 \|y_d\|_2 \geq \frac{1}{4 \pi} \frac{1}{2^d} \|x\|_2 \|y\|_2,
\]
which establishes \eqref{piwxpiwyy-lowerbound}.

\textit{Part III: Establishing \eqref{vxyhxy-d-cubed}}\\
This proof will proceed by setting up and solving a recurrence relation.  We will use that  
\begin{align}
\Gamma_d = s_d \Gamma_{d-1} + r_d, \ \Gamma_0 = y \quad \Rightarrow \quad \Gamma_d = \Bigl(\prod_{i=1}^d s_i \Bigr) y + \sum_{i=1}^d \Bigr( r_i \prod_{j = i+1}^d s_j \Bigr). \label{recurrence-relation}
\end{align}

First, we derive a recurrence relation and solve for $(\PiWd)^t (\PiWd)$.  We have
\begin{align*}
M_d:=  (\PiWd)^t (\PiWd) 
&= (\PiWn{d-1})^t \Bigl(\frac{1}{2}I_{n_{d-1}} + O_1(\eps) \Bigr)  (\PiWn{d-1}) \\
&= \frac{1}{2} (\PiWn{d-1})^t (\PiWn{d-1})  + O_1(\eps \Pi_{i=1}^{d-1} \| \Wipx\|^2) \\
&= \frac{1}{2} M_{d-1} + O_1 \Bigl(\eps \frac{1 + 4 \eps (d-1)}{2^{d-1}} \Bigr),
\end{align*}
where the first equality follows by \eqref{defn-E1-symmetric-case}, and the third equality follows from \eqref{bound-piwixiy-on-E1}, as $2 \eps d \leq 1$.
Solving this recurrence relation with $M_0 = I_{n_0}$ by \eqref{recurrence-relation}, we get that if $4 d \eps \leq 1$, then 
\begin{align}
(\PiWd)^t (\PiWd)  &= \frac{1}{2^d} I_{n_0} + \sum_{i=1}^d O_1\Bigl(\eps \frac{1 + 4 \eps (i-1)}{2^{i-1}} \Bigr) \frac{1}{2^{d-i}} \notag \\
&= \frac{1}{2^d}I_{n_0} + \frac{4 \eps d}{2^d}  O_1(1).  \label{error-bound-md}
\end{align}
Thus, 
\begin{align}
( \PiWd)^t ( \PiWd ) x = \frac{1}{2^d} x + O_1\Bigl( \frac{4 \eps d}{2^d} \Bigr) \|x\|_2 \label{vxy-first-term-d-layer-estimate}
\end{align}

Next, we  derive a recurrence relation for $\Gamma_d := (\PiWd)^t (\PiWdy) y$.  
\begin{align}
\Gamma_{d} &= (\PiWn{d-1})^t (\Wnpx{d}^t \Wnpy{d}) (\PiWny{d-1}) y\\
&= (\PiWn{d-1})^t  \left( \frac{\pi-\theta_{d-1}}{2\pi}I_{n_{d-1}} + \frac{\sin \theta_{d-1}}{2\pi} M_{\hat{x}_{d-1} \leftrightarrow \hat{y}_{d-1} }  + O_1(\eps) \right) \left( \Pi_{i=d-1}^1 W_{i,+,y} \right) (y) \notag \\
&= \frac{\pi-\theta_{d-1}}{2\pi} \Gamma_{d-1}  +  \frac{\sin \theta_{d-1}}{2\pi} \frac{\|y_{d-1}\|_2}{\|x_{d-1}\|_2} (\PiWn{d-1})^t (\PiWn{d-1}) x + \eps \Bigl( \frac{1 + 4 \eps d}{2^{d-1}} \Bigr) \|y\|_2 O_1(1) \notag\\
&= \frac{\pi-\theta_{d-1}}{2\pi} \Gamma_{d-1}  +  \frac{\sin \theta_{d-1}}{2\pi} \frac{\|y_{d-1}\|_2}{\|x_{d-1}\|_2} \frac{x}{2^{d-1}} + \frac{1}{2\pi} \frac{\|y_{d-1}\|_2}{\|x_{d-1}\|_2}  \frac{4 d\eps}{2^{d-1}} \|x\|_2 O_1(1) + \eps \Bigl( \frac{1 + 4 \eps d}{2^{d-1}} \Bigr) \|y\|_2 O_1(1) \notag\\
&= \frac{\pi-\theta_{d-1}}{2\pi} \Gamma_{d-1}  +  \frac{\sin \theta_{d-1}}{\pi} \frac{\|y_{d-1}\|_2}{\|x_{d-1}\|_2} \frac{x}{2^d} + \frac{1+8d\eps}{2\pi}  \frac{4 d\eps}{2^{d-1}} \|y\|_2 O_1(1) + \eps \Bigl( \frac{2}{2^{d-1}} \Bigr) \|y\|_2 O_1(1) \notag\\
&= \frac{\pi-\theta_{d-1}}{2\pi} \Gamma_{d-1}  +  \frac{\sin \theta_{d-1}}{\pi} \frac{\|y_{d-1}\|_2}{\|x_{d-1}\|_2} \frac{x}{2^d} + \eps \Bigl(\frac{3}{2\pi} \frac{2 \cdot 4 d }{2^d} +  \frac{4}{2^d} \Bigr) \|y\|_2  O_1(1) \notag\\
&= \frac{\pi-\theta_{d-1}}{2\pi} \Gamma_{d-1}  +  \frac{\sin \theta_{d-1}}{\pi} \frac{\|y_{d-1}\|_2}{\|x_{d-1}\|_2} \frac{x}{2^d} + \frac{8}{2^d} d \eps \|y\|_2 O_1(1) \notag\\
&= \frac{\pi-\theta_{d-1}}{2\pi} \Gamma_{d-1}  +  \frac{\sin \theta_{d-1}}{\pi} \frac{\|y\|_2}{\|x\|_2} \frac{x}{2^d} + \frac{1}{2^d} \frac{8 d \eps}{ \pi} \|y\|_2 O_1(1) + \frac{8}{2^d} d \eps \|y\|_2 O_1(1)\notag \\
&= \frac{\pi-\theta_{d-1}}{2\pi} \Gamma_{d-1}  +  \frac{\sin \theta_{d-1}}{\pi} \frac{\|y\|_2}{\|x\|_2} \frac{x}{2^d} + 11 d \eps \frac{\|y\|_2}{2^d} O_1(1) \notag\\
&= \Bigl( \frac{\pi-\thetabar_{d-1}}{2\pi} +O_1 \bigl(\frac{d \delta}{2 \pi} \bigr) \Bigr) \Gamma_{d-1}  +  \frac{\sin \thetabar_{d-1}}{\pi} \frac{\|y\|_2}{\|x\|_2} \frac{x}{2^d} + 2 d \delta \frac{\|y\|_2}{2^d} O_1(1)  \label{recurrence-relation-gammad}
\end{align}
where the second line follows from \eqref{defn-E1}, the definition $\Wnpy{d}= (W_d)_{+, y_{d-1}}$, and the definition of $Q_{x_d, y_d}$ in \eqref{WDC}; the third line follows by the definition of $M_{\hat{x}_{d-1} \leftrightarrow \hat{y}_{d-1} }$ in \eqref{WDC}, the bound \eqref{bound-piwixiy-on-E1}, and the definition of $x_{d-1}$; the fourth line uses \eqref{error-bound-md}; the fifth line follows from $4 \eps d \leq 1$ and \eqref{bound-yd-over-xd-on-E1}; the eighth line follows from \eqref{bound-yd-over-xd-on-E1}; and the last line follows from \eqref{bound-thetad-thetadbar}, $ \delta = 4 \sqrt{\eps}$, and $11 \sqrt{\eps} \leq 4$.

Solving the recurrence relation \eqref{recurrence-relation-gammad} using \eqref{recurrence-relation}, we get that
\begin{align}
\Gamma_d = &\bigPioned \Bigl[\frac{\pi - \thetabar_{i-1}}{2\pi} + O_1\bigl(\frac{i \delta}{2 \pi} \bigr) \Bigr] y  \notag \\ & \ \ +
\sum_{i=1}^d  \Bigl( \frac{\sin \thetabar_{i-1}}{ \pi} \frac{\|y\|_2}{\|x\|_2} \frac{x}{2^i} + \frac{2 i \delta}{2^i} \|y\|_2 O_1(1) \Bigr) \prod_{j=i+1}^d \Bigl( \frac{\pi - \thetabar_{j-1}}{2\pi} + O_1\bigl(\frac{j \delta}{2\pi} \bigr)    \Bigr) \label{soln-recurrence-Gammad}
\end{align}
First, we control the first term of $\Gamma_d$ in \eqref{soln-recurrence-Gammad}.  We have
\begin{align}
\left |   \bigPioned \Bigl[\frac{\pi - \thetabar_{i-1}}{2\pi} + O_1 \bigl(\frac{i \delta}{2 \pi} \bigr) \Bigr] y - \bigPioned \Bigl[ \frac{\pi - \thetabar_{i-1}}{2\pi} \Bigr]  y  \right | & \leq \Bigl( \bigPioned \Bigl[\frac{1}{2} + \frac{i \delta}{2\pi}  \Bigr] - \frac{1}{2^d} \Bigr) \|y\|_2 \notag\\
&\leq \frac{1}{2^d} \Bigl[ \Bigl(1 + \frac{d \delta}{\pi} \Bigr)^d - 1  \Bigr] \|y\|_2 \notag \\
&\leq \frac{1}{2^d} \Bigl[ e^{d \log(1+d \delta / \pi)} - 1    \Bigr] \|y\|_2 \notag\\
&\leq \frac{1}{2^d} \Bigl[ e^{d^2 \delta/\pi} - 1    \Bigr] \|y\|_2 \notag\\
&\leq \frac{2}{2^d} \frac{d^2 \delta}{\pi} \|y\|_2 \leq \frac{d^2 \delta}{2^d} \|y\|_2, \label{estimate-gammad-first-term}
\end{align}
where the first inequality follows by Lemma \ref{lemma:piritiri} and as $\frac{\pi - \thetabar_i}{2\pi} \in [0, 1/2]$; and fifth inequality follows because $e^{d^2 \delta / \pi} \leq  1+ 2 \frac{d^2 \delta}{\pi}$ if $d^2 \delta / \pi \leq 1$.

Next, we control the second term of $\Gamma_d$  in \eqref{soln-recurrence-Gammad}.
Similar to the calculation above, we have
\begin{align}
\left |   \prod_{j=i+1}^d \Bigl[\frac{\pi - \thetabar_{j-1}}{2\pi} + O_1(\frac{j \delta}{2 \pi}) \Bigr] - \prod_{j=i+1}^d \frac{\pi - \thetabar_{j-1}}{2\pi}    \right | & \leq \frac{d^2 \delta}{2^{d-i}} \label{second-term-partial-estimate}
\end{align}
if $d^2 \delta/\pi \leq 1$.   We now get that the second term of \eqref{soln-recurrence-Gammad} is
\begin{align}
&\sum_{i=1}^d \Bigl[ \frac{\sin \thetabar_{i-1}}{ \pi} \frac{\|y\|_2}{\|x\|_2} \frac{x}{2^i} + \frac{2 i \delta}{2^i} \|y\|_2 O_1(1) \Bigr] \prod_{j=i+1}^d \Biggl( \frac{\pi - \thetabar_{j-1}}{2\pi} + O_1 \Bigl(\frac{j \delta}{2\pi} \Bigr)    \Biggr) \notag\\
&= \sum_{i=1}^d \Bigl[ \frac{\sin \thetabar_{i-1}}{ \pi} \frac{\|y\|_2}{\|x\|_2} \frac{x}{2^i} + \frac{2 i \delta}{2^i} \|y\|_2 O_1(1) \Bigr] \Bigl[ \Bigl( \prod_{j=i+1}^d  \frac{\pi - \thetabar_{j-1}}{2\pi} \Bigr) +O_1\Bigl(\frac{d^2 \delta}{2^{d-i}}    \Bigr) \Bigr] \notag \\
&= \left[  \sum_{i=1}^d \frac{\sin\thetabar_{i-1}}{\pi} \frac{\|y\|_2}{\|x\|_2} \frac{x}{2^d} \prod_{j=i+1}^d \frac{\pi - \thetabar_{j-1}}{\pi} \right] + O_1\Bigl( \frac{5 d^3 \delta}{2^d}\Bigr) \|y\|_2 \label{estimate-gammad-second-term}
\end{align}
where the first equality follows by \eqref{second-term-partial-estimate} and $d^2 \delta / \pi \leq 1$, and the second equality follows by expanding the terms and using $d^2 \delta \leq 1$.

Combining \eqref{estimate-gammad-first-term} and \eqref{estimate-gammad-second-term}, we get
\begin{align*}
\Gamma_d = \frac{1}{2^d} \left[ \bigPioned \frac{\pi - \thetabar_{i-1}}{\pi} \right] y + \frac{1}{2^d} \left[ \sum_{i=1}^d \frac{\sin \thetabar_{i-1}}{\pi} \Bigl (\prod_{j=i+1}^d \frac{\pi - \thetabar_{j-1}}{\pi} \Bigr) \frac{\|y\|_2}{\|x\|_2} x \right ] + O_1\Bigl(\frac{6 d^3 \delta}{2^d} \Bigr) \|y\|_2. 
\end{align*}
We complete the proof of \eqref{vxyhxy-d-cubed} by combining this equality with \eqref{vxy-first-term-d-layer-estimate}, and $\delta = 4 \sqrt{\eps}$.

\end{proof}

In the proof of the previous lemma, we used the following technical result.  

\begin{lemma} \label{lemma:piritiri} 
Let $d \in \mathbb{N}$ and let $0 \leq r_i \leq \rmax$ for $i = 1 \ldots d$. We have
\begin{align*}
\left |  \prod_{i=1}^d (r_i + t_i)  -  \prod_{i=1}^d r_i \right |  \leq  \prod_{i=1}^d (\rmax + |t_i|) - \rmax^d.
\end{align*}
\end{lemma}
\begin{proof}
First, we establish that 
\begin{align*}
\Bigl | \bigPioned (r_i + t_i) - \bigPioned r_i \Bigr | \leq \bigPioned (r_i+ |t_i|) - \bigPioned r_i.
\end{align*}
This follows by noting that for $\delta \in [0,1]$, 
\begin{align*}
\left | \frac{d}{d\delta} \Bigl[ \prod_{i=1}^d (r_i + \delta t_i)   \Bigr]    \right| &= \left| \sum_{i=1}^d t_i  \prod_{j \neq i} (r_j + \delta t_j)    \right | 
\leq \sum_{i=1}^d |t_i| \prod_{j \neq i}  (r_j + \delta |t_j|)  
= \frac{d}{d \delta} \Bigl( \prod_{i=1}^d (r_i + \delta |t_i|)  \Bigr),
\end{align*}
and integrating over $\delta \in [0,1]$.  
Next, we establish that
\begin{align*}
\bigPioned (r_i+ |t_i|) - \bigPioned r_i \leq \bigPioned (\rmax + |t_i|) - \rmax^d.
\end{align*}
This follows by noting that for all $k = 1 \ldots d$,
\begin{align*}
\frac{\partial}{\partial r_k} \Bigl[ \bigPioned (r_i + |t_i|) - \bigPioned r_i \Bigr] = \prod_{j \neq k} (r_j + |t_j|) - \prod_{j \neq k} r_j \geq 0.
\end{align*}
\end{proof}



\subsection{Proof of Deterministic Theorem} \label{sec:deterministic-theorem}

The proof of Theorem \ref{thm-multi-layer-deterministic} follows the outline provided in Section \ref{sec:proofs}.  

\begin{proof}[Proof of Theorem \ref{thm-multi-layer-deterministic}]
Recall that

\begin{align*}
\vxxo &= \begin{cases}\nabla f(x) & G \text{ is differentiable at x,}\\ \lim_{\delta \to 0^+} \nabla f(x + \delta w) & \text{otherwise,} \end{cases} \end{align*}
where $G$ is differentiable at $x+\delta w$ for sufficiently small $\delta$.  Such a $w$ exists by the piecewise linearity of $G$, and any such $w$ can be selected arbitrarily.  Also, recall that
\begin{align*}
 \nabla f(x) = (\PiWd)^t \AtA (\PiWd) x -  (\PiWd)^t \AtA (\PiWdo) \xo.
\end{align*}

Let
\begin{align*}
\vbarxxo &= (\PiWd)^t  (\PiWd) x -  (\PiWd)^t (\PiWdo) \xo,\\
\hxxo 
&=  - \frac{1}{2^d} \Bigl( \prod_{i=0}^{d-1} \frac{\pi - \thetabar_i}{\pi}  \Bigr)\xo 
+ \frac{1}{2^d} \left[ x - \sum_{i=0}^{d-1} \frac{\sin \thetabar_i}{\pi}  \Bigl( \prod_{j=i+1}^{d-1} \frac{\pi - \thetabar_j}{\pi}  \Bigr)  \frac{\|\xo\|_2}{\|x\|_2} x  \right], \\
S_{\eps, \xo} &= \Bigl \{ x \in \R^k \mid \| \hxxo \|_2 \leq \frac{1}{2^d} \eps \max(\|x\|_2, \|\xo\|_2 ) \Bigr \},
\end{align*}
where $\thetabar_i = g(\thetabar_{i-1})$ and $\thetabar_0 = \angle(x, \xo)$. 
For brevity of notation, write $\vx = \vxxo, \vbarx = \vbarxxo,$ and $\hx = \hxxo$.  

The WDC implies that  for all $x \neq 0$ and for all $i = 1 \ldots d$, 
\begin{align}
\| \Wipx\|^2 \leq \frac{1}{2} + \eps. \label{bound-Wiplusx-squared}
\end{align}

Now, we establish that for all differentiable points $x \in \R^k$,
\begin{align}
\| \nabla f(x) - \vbarx  \| \leq 2 \eps \Bigl(\frac{1}{2} + \eps \Bigr)^d  \max(\|x\|_2, \|\xo\|_2). \label{isometry-condition-deterministic-goal}
\end{align}
At $x \in \R^k$ such that $G$ is differentiable at $x$, the local linearity of $G$ gives that $G(x+z) - G(x) = (\PiWd) z$ for any sufficiently small $z \in \R^k$.   By the RRIC, we have 
\begin{align}
| \langle \red{A} \PiWd z, \red{A} \PiWdy \ztilde \rangle  - \langle  \PiWd z,  \PiWdy \ztilde \rangle | \leq \eps \Pioned \| \Wipx\| \|\Wipy\| \|z\|_2 \|\ztilde\|_2. \label{RRIC-application}
\end{align}
for all $z, \ztilde$, which, together with \eqref{bound-Wiplusx-squared}, implies \eqref{isometry-condition-deterministic-goal}.

Similarly, the WDC implies by Lemma \ref{lemma:Gradient-multilayer-concentration} and $\eps < 1/(16 \pi d^2)^2$ that for all nonzero $x, \xo, y \in \R^k$, 
\begin{align}
\| \vbarx - \hx \|_2 &\leq K \frac{d^3 \sqrt{\eps}}{2^d} \max ( \|x\|_2, \| \xo\|_2), \text{ and } \label{vxbarhx-are-close}\\
\bigl \langle  (\PiWd) x, ( \PiWdy) y \bigr \rangle  &\geq \frac{1}{4 \pi} \frac{1}{2^d} \|x\|_2 \|y\|_2. \label{piwipxxpiwipyy}
\end{align}

Thus, we have, for all $x\neq 0, \xo \neq 0$, 
\begin{align}
\|\vx - \hx\|_2  &=\lim_{\delta \to 0^+} \|  \nabla  f(x + \delta w) - h_{x + \delta w}  \|_2 \notag \\
& \leq \lim_{\delta \to 0^+}  \bigl( \|\nabla f(x + \delta w)- \vbar_{x + \delta w}\|_2 + \|\vbar_{x + \delta w} - h_{x + \delta w}\|_2 \bigr) \notag\\
&\leq \sqrt{\eps} \Bigl(2 \frac{(1+2 \eps)^d}{2^d} + K \frac{d^3}{2^d} \Bigr) \max(\|x\|_2, \|\xo\|_2) \notag \\
&\leq \sqrt{\eps} \Ktilde \frac{d^3}{2^d}  \max(\|x\|_2, \|\xo\|_2), \label{vxgx-approx-highd}
\end{align}
for some universal constant $\Ktilde$, where the first inequality follows by the definition of $\vx$ and the continuity of $\hx$ for nonzero $x$; the second inequality follows by combining \eqref{bound-Wiplusx-squared}, \eqref{isometry-condition-deterministic-goal},  \eqref{vxbarhx-are-close},  and as $2 d \eps \leq 1 \Rightarrow (1+2 \eps)^d \leq e^{2 \eps d} \leq 1 + 4\eps d$.

Note that the one-sided directional derivative of $f$ in the direction of $y\neq 0$ at $x$ is $D_{y} f(x) = \lim_{t \to 0^+} \frac{f(x+t y) - f(x)}{t}$.  Due to the continuity and piecewise linearity of the function \[G(x) = \relu(\Win{d} \ldots \relu (\Win{2} \relu(\Win{1} x))  ),\] we have that for any $x, y \neq 0$ that there exists a sequence $\{x_n\} \to x$ such that $f$ is differentiable at each $x_n$ and $D_y f(x) = \lim_{n \to \infty} \nabla f(x_n)\cdot y$. Thus, as $\nabla f(x_n) = v_{x_n}$, 
\[
D_{-\vx} f(x)= - \lim_{n \to \infty} \vxn \cdot \red{\frac{\vx}{\|\vx\|}}.
\]
Now, we write
\begin{align*}
\vxn \cdot \vx &= \hxn \cdot \hx + (\vxn - \hxn) \cdot \hx + \hxn \cdot (\vx - \hx) + (\vxn - \hxn) \cdot (\vx - \hx)\\
&\geq \hxn \cdot \hx  - \| \vxn - \hxn\|_2 \| \hx\|_2  - \| \hxn \|_2 \|\vx - \hx\|_2 - \| \vxn - \hxn\|_2 \|\vx - \hx\|_2 \\
&\geq \hxn \cdot \hx -  \sqrt{\eps}  \Ktilde \frac{d^3}{2^d}  \max(\| \xn\|_2, \|\xo\|_2) \| \hx\|_2  -  \sqrt{\eps} \Ktilde \frac{d^3}{2^d} \max(\|x\|_2, \|\xo\|_2) \| \hxn\|_2 \\ & \ \ \ -  \eps  \Bigl(\Ktilde \frac{d^3}{2^d} \Bigr)^2 \max(\|\xn\|_2, \|\xo\|_2)  \max(\|x\|_2, \|\xo\|_2),
\end{align*}
where the first inequality follows from the triangle inequality, and the second inequality follows by \eqref{vxgx-approx-highd}.
As $\hx$ is continuous in $x$ for all nonzero $x$, we have for any $x \in S_{\red{8 \sqrt{\eps}} \Ktilde d^3 }^c$, 
\begin{align*}
\lim_{n \to \infty} \vxn \cdot \vx &\geq \| \hx\|_2^2 - 2 \red{\sqrt{\eps}} \Ktilde \frac{d^3}{2^d} \| \hx\|_2 \max(\|x\|_2, \|\xo\|_2) -  \eps \Bigl [ \Ktilde \frac{d^3}{2^d} \Bigr]^2 \max(\|x\|_2, \|\xo\|_2)^2 \\
&= \frac{\|\hx\|_2}{2} \Bigl [ \| \hx\|_2 - 4 \sqrt{\eps} \Bigl ( \Ktilde \frac{d^3}{2^d} \Bigr) \max(\|x\|_2, \|\xo\|_2) \Bigr]  + \frac{1}{2} \Bigl[ \| \hx \|^2 -2 \cdot  \eps \Bigl ( \Ktilde \frac{d^3}{2^d} \Bigr)^2 \max(\|x\|_2, \|\xo\|_2 )^2   \Bigr] \\
&\geq \red{  \frac{\|\hx\|_2}{2}  4 \sqrt{\eps} \Bigl ( \Ktilde \frac{d^3}{2^d} \Bigr) \max(\|x\|_2, \|\xo\|_2)    } \\
 &\geq \red{ \frac{7/8 \|\vx\|_2}{2} 4 \sqrt{\eps} \Bigl ( \Ktilde \frac{d^3}{2^d} \Bigr) \max(\|x\|_2, \|\xo\|_2), }\\
\end{align*}
\red{where the last inequality uses \eqref{vxgx-approx-highd} and the definition of $S_{\red{8 \sqrt{\eps}} \Ktilde d^3 }$.}
We conclude $D_{-\vx} f(x) < 0$ for all nonzero $x \in S_{\red{8} \sqrt{\eps} \Ktilde d^3}^c$.

It remains to prove that $\forall x \neq 0$,  $D_x f(0) < 0$.  
We compute that
\begin{align*}
D_x f(0) \cdot \red{\|x\|_2}&= -\langle A (\PiWd) x, A (\PiWdo) \xo \rangle \\
& = - \langle x, (\PiWd)^t \AtA (\PiWdo) \xo \rangle \\
& = - \langle x, (\PiWd)^t (\AtA - I_{n_2}) (\PiWdo) \xo \rangle \notag \\
& \quad - \langle (\PiWd) x,  (\PiWdo) \xo \rangle\\
& \leq \eps \frac{(1 + 2 \eps)^d}{2^d} \|x\|_2 \| \xo \|_2 - \frac{1/(4\pi)}{2^d} \| x\|_2 \| \xo\|_2 \\
& \leq  \frac{2 \eps}{2^d} \|x\|_2 \| \xo \|_2 - \frac{1/(4\pi)}{2^d} \| x\|_2 \| \xo\|_2 
\end{align*}
where the first inequality holds by \eqref{RRIC-application}, \eqref{bound-Wiplusx-squared}, and \eqref{piwipxxpiwipyy}; and the second inequality follows from $4 \eps d \leq 1$. Thus, for $\eps < \red{\frac{1}{16\pi}}$, $D_x f(0) < \red{-\frac{1}{8 \pi 2^d} \|\xo\|_2}$. 

The proof is finished by applying Lemma \ref{lemma:Seps} and $8 \pi d^6 \sqrt{\red{8} \sqrt{\eps} \Ktilde d^3} \leq 1$ to get 
\begin{align*}
S_{\red{8} \sqrt{\eps} \Ktilde d^3} \subset \mathcal{B}(x_0, 56 d \sqrt{\red{8} \sqrt{\eps} \Ktilde d^3} \|x_0\|_2  ) \cup  \mathcal{B}( -\zetacheck_d x_0, 500 d^{11} \sqrt{\red{8} \sqrt{\eps} \Ktilde d^3} \|\xo\|_2).
\end{align*}
\end{proof}

\subsection{Control of the zeros of $\hxxo$} \label{sec:zeros-hxxo}

We now show that $\hxxo$ is away from zero outside of a neighborhood of $\xo$ and $-\rho_d \xo$.

\begin{lemma} \label{lemma:Seps}
Suppose $8 \pi d^6 \sqrt{\eps} \leq 1$.  Let 
\[
S_{\eps,x_0}  = \left \{ x \neq 0 \in \R^k  \mid \|\hxxo \|_2 \leq \frac{1}{2^d}\eps \max \left( \|x\|_2, \|x_0\|_2\right) \right\},
\]
where $d$ is an integer greater than $1$ and let $\hxxo$ be defined by

\begin{align}
\frac{\hxxo}{\|\xo\|_2}&=   - \frac{1}{2^d} \Bigl( \prod_{i=0}^{d-1} \frac{\pi - \thetabar_i}{\pi}  \Bigr) \xohat + \frac{1}{2^d}  \Bigl[  \frac{\|x\|_2}{\|x_0\|_2} - \sum_{i=0}^{d-1} \frac{\sin \thetabar_{i}}{\pi} \prod_{j=i+1}^{d-1} \frac{\pi - \thetabar_{j}}{\pi}   \Bigr]  \hat{x},   
\end{align}
where $\thetabar_0 = \angle(x, \xo)$ and $\thetabar_i = g(\thetabar_{i-1})$ for $g$ given by \eqref{defn-g}. 
Define
\[
\zetacheck_d := \sum_{i=0}^{d-1} \frac{\sin \thetacheck_{i}}{\pi} \left( \prod_{j=i+1}^{d-1} \frac{\pi - \thetacheck_{j} }{\pi} \right),
\]
where $\thetacheck_0= \pi$ and $\thetacheck_i = g(\thetacheck_{i-1})$.  If $x \in \Sepsxo$, then we have that either
\[
|\thetabar_0| \leq 2 \sqrt{\epsilon} \quad \text{and} \quad |\|x\|_2 - \|x_0\|_2| \leq  18 d \sqrt{\epsilon} \|x_0\|_2
\]
or
\[
|\thetabar_0 - \pi| \leq  8\pi d^4 \sqrt{\epsilon}  \quad \text{and} \quad  \left| \|x\|_2 - \|x_0\|_2 \zetacheck_d \right | \leq 200 d^7 \sqrt{\epsilon} \|x_0\|_2.
\]
In particular, we have
\begin{align}
\Sepsxo \subset \mathcal{B}(x_0, 56 d \sqrt{\eps} \|x_0\|_2  ) \cup  \mathcal{B}( -\zetacheck_d x_0, 500 d^{11} \sqrt{\eps} \|\xo\|_2).
 \label{Seps-estimate-balls}
\end{align}
Additionally, $\rho_d \to 1$ as $d \to \infty$.
\end{lemma}

\begin{proof}
Without loss of generality, let $\|\xo\|_2=1$, $\xo = e_1$ and $\red{x} = r \cos \thetabar_0 \cdot e_1 + r \sin \thetabar_0 \cdot e_2$ for $\thetabar_0 \in [0, \pi]$.  Let $x \in \Sepsxo$. 

First we introduce some notation for convenience.  Let
\[
\xi = \Pipithetapii, \quad \zeta = \zetaterm, \quad r = \|x\|_2, \quad M = \max(r, 1).
\]
Thus, $\hxxo = - \red{\frac{1}{2^d}}\xi  \xohat + \red{\frac{1}{2^d}} (r- \zeta) \xhat$. 
By inspecting the components of $\red{\hxxo}$, we have that $x\in \Sepsxo$ implies
\begin{align}
|-\xi + \cos \thetabar_0 ( r - \zeta)| \leq \eps M \label{Seps-e1term}\\
|\sin \thetabar_0 (r - \zeta)| \leq \eps M \label{Seps-e2term}
\end{align}

Now, we record several properties.  We have:
\begin{align}
\thetabar_{i} &\in [0, \pi/2] \text{ for } i \geq \red{1} \notag \\
\thetabar_{i} &\leq \thetabar_{i-1} \text{ for } i \geq 1 \notag\\
 |\xi| &\leq 1 \label{bound-c-one}\\
| \zeta| &\leq \red{\min \Bigl(\frac{d}{\pi}, \frac{d}{\pi} \thetabar_0 \Bigr)} \label{zeta-bound-d-sin-theta} \\
\thetacheck_i &\leq \frac{3\pi}{i+3} \text{ for $i \geq 0$} \label{upper-bound-thetai}\\
\thetacheck_i &\geq \frac{\pi}{i+1} \text{ for $i \geq 0$} \label{lower-bound-thetai}\\
\xi = \prod_{i=0}^{d-1} \frac{\pi - \thetabar_i}{\pi} &\geq \frac{\pi - \thetabar_0}{\pi } d^{-3} \label{pi-theta-pi-d-cubed}\\
\thetabar_0 = \pi + O_1(\delta) &\Rightarrow \thetabar_i = \thetacheck_i + O_1(i \delta) \label{thetabar-thetacheck-bound}\\
\thetabar_0 = \pi + O_1(\delta) &\Rightarrow |\xi| \leq \frac{\delta}{\pi} \label{c-one-bound-thetabarzero}\\
\thetabar_0 = \pi + O_1(\delta) &\Rightarrow \zeta = \zetacheck_d + O_1(3 d^3 \delta) \text{ if } \frac{d^2 \delta}{\pi} \leq 1 \label{high-theta-zeta-zetabar-control}
\end{align}
We now establish \eqref{upper-bound-thetai}.  Observe $0< g(\theta) \leq \bigl( \frac{1}{3 \pi} + \frac{1}{\theta} \bigr)^{-1} =: \gtilde(\theta)$ for $\theta \in (0, \pi]$.  As $g$ and $\gtilde$ are monotonic increasing, we have $\thetacheck_i = g^{\circ i}(\thetacheck_0) = g^{\circ i}(\pi) \leq \gtilde^{\circ i}(\pi) = \bigl( \frac{i}{3 \pi} + \frac{1}{\pi} \bigr)^{-1} = \frac{3\pi}{i + 3}$.  Similarly, $g(\theta) \geq (\frac{1}{\pi} + \frac{1}{\theta})^{-1}$ implies that $\thetacheck_i \geq \frac{\pi}{i+1}$, establishing \eqref{lower-bound-thetai}. 

We now establish \eqref{pi-theta-pi-d-cubed}.  Using \eqref{upper-bound-thetai} and $\thetabar_i \leq \thetacheck_i$, we have
\begin{align*}
\prod_{i=1}^{d-1} \Bigl(1 - \frac{\thetabar_i}{\pi} \Bigr) &\geq \prod_{i=1}^{d-1} \Bigl(1 - \frac{3}{i+3}  \Bigr)
\geq d^{-3},
\end{align*}
where the last inequality can be established by showing that the ratio of consecutive terms with respect to $d$ is greater for the product in the middle expression than for $d^{-3}$.

We establish \eqref{thetabar-thetacheck-bound} by using the fact that $|g'(\theta)| \leq 1$ for all $\theta \in [0, \pi]$ and using the same logic as for \eqref{bound-thetad-thetadbar}.

We now establish \eqref{high-theta-zeta-zetabar-control}.  As $\thetabar_0 = \pi + O_1(\delta)$, we have $\thetabar_i = \thetacheck_i + O_1(i \delta)$.  Thus, if $\frac{d^2 \delta}{\pi}\leq 1$,
\[
\prod_{j=i+1}^{d-1} \frac{\pi - \thetabar_j}{\pi} = \prod_{j=i+1}^{d-1} \Bigl( \frac{\pi - \thetacheck_j}{\pi} + O_1(\frac{i \delta}{2 \pi} ) \Bigr) = \Bigl(\prod_{j=i+1}^{d-1}  \frac{\pi - \thetacheck_j}{\pi} \Bigr)  + O_1(d^2 \delta )
\]
So
\begin{align}
\zeta &= \sum_{i=0}^{d-1} \Bigl( \frac{\sin \thetacheck_i}{\pi} + O_1(\frac{i \delta}{\pi}) \Bigr)  \Bigl[ \Bigl(\prod_{j=i+1}^{d-1} \frac{\pi - \thetacheck_j}{\pi} \Bigr) +O_1(d^2 \delta)  \Bigr] \\
&= \zetacheck_d + O_1 \Bigl( d^2\delta / \pi + d^3 \delta /\pi + d^4 \delta^2/\pi   \Bigr) \\
&= \zetacheck_d + O_1( 3 d^3 \delta).
\end{align}
Thus \eqref{high-theta-zeta-zetabar-control} holds.

Next, we establish that $x \in \Sepsxo \Rightarrow r \leq 4d$, and thus $M \leq 4d$.  Suppose $r > 1$.  At least one of the following holds: $|\sin\thetabar_0| \geq 1/\sqrt{2}$ or $|\cos \thetabar_0| \geq 1/\sqrt{2}$. If $|\sin \thetabar_0| \geq 1/\sqrt{2}$ then  \eqref{Seps-e2term} implies that $|r - \zeta| \leq \sqrt{2} \eps r$.  Using \eqref{zeta-bound-d-sin-theta}, we get $ r \leq \frac{d/\pi}{1 - \sqrt{2} \eps} \leq d/2$ if $\eps < 1/4$.  If $| \cos \thetabar_0| \geq 1/\sqrt{2}$, then \eqref{Seps-e1term} implies that 
$|r - \zeta| \leq \sqrt{2} (\eps r + |\xi|)$.   Using \eqref{bound-c-one}, \eqref{zeta-bound-d-sin-theta}, and $\eps<1/4$, we get $r \leq \frac{\sqrt{2} |\xi| +  \zeta}{1 - \sqrt{2}{\eps}} \leq \frac{\sqrt{2} + d }{1 - \sqrt{2} \eps} \leq 4 d$. Thus, we have $x \in S_\eps \Rightarrow r \leq 4d \Rightarrow M \leq 4d$. 

Next, we establish that we only need to consider the small angle case ($\thetabar_0 \approx 0$) and the large angle case ($\thetabar_0 \approx \pi$).  Exactly one of the following holds: $|r - \zeta| \geq \sqrt{\eps} M$ or $|r -\zeta| < \sqrt{\eps} M$.    If $|r - \zeta| \geq \sqrt{\eps}M$, then by \eqref{Seps-e2term}, we have $|\sin\thetabar_0| \leq \sqrt{\eps}$.  Hence $\thetabar_0 = O_1(2 \sqrt{\eps})$ or $\thetabar_0 = \pi + O_1(2 \sqrt{\eps})$, as $\eps < 1$. If $|r - \zeta| \leq \sqrt{\eps} M$, then by \eqref{Seps-e1term} we have $|\xi| \leq 2 \sqrt{\eps} M$.  Using \eqref{pi-theta-pi-d-cubed}, we get $\thetabar_0 = \pi + O_1(2 \pi d^3 \sqrt{\eps} M)$.  Thus, we only need to consider the small angle case, $\thetabar_0 =O_1(2 \sqrt{\eps})$ and the large angle case $\thetabar_0 = \pi + O_1(8 \pi d^4 \sqrt{\eps})$, where we have used $M \leq 4 d$.

\textbf{Small Angle Case}.  Assume $\thetabar_0 = O_1(2 \sqrt{\eps})$.  As $\thetabar_i \leq \thetabar_0 \leq 2 \sqrt{\eps}$ for all $i$, we have $\xi \geq (1 - \frac{2\sqrt{\eps}}{\pi})^d = 1 + O_1(\frac{4 d \sqrt{\eps}}{\pi})$ provided $2d \sqrt{\eps} \leq 1/2$.   By \eqref{zeta-bound-d-sin-theta}, we also have $\zeta = O_1( \frac{d}{\pi} 2 \sqrt{\eps}) = O_1(d\sqrt{\eps})$.  By \eqref{Seps-e1term},  we have
\[
|-\xi + \cos \thetabar_0 ( r - \zeta)| \leq \eps M.
\]
Thus, as $\cos \thetabar_0 = 1 + O_1(\thetabar_0^2/2) = 1 + O_1(2 \eps)$,
\[
-\Bigl( 1 + O_1(4 d \sqrt{\eps}) \Bigr) + (1 + O_1(2 \eps) )(r + O_1( d \sqrt{\eps})) = O_1(4 d \eps),
\]
and thus,
\begin{align}
r-1 &= O_1(4d \sqrt{\eps} + 2 \eps 4 d +  d \sqrt{\eps} + 2 d \eps^{3/2} + 4 \eps d)\\
&= O_1(18d \sqrt{\eps}).
\end{align}

\textbf{Large Angle Case}.  Assume $\theta_0 = \pi + O_1(\delta)$ where $\delta = 8 \pi d^4 \sqrt{\eps}$.  By \eqref{c-one-bound-thetabarzero} and \eqref{high-theta-zeta-zetabar-control}, we have $\xi = O_1(\delta/\pi)$, and we have $\zeta = \zetacheck_d + O_1(3 d^3 \delta)$ if $8 d^6 \sqrt{\eps} \leq 1$.
By \eqref{Seps-e1term}, we have 
\[
|-\xi +  \cos \theta_0 (r - \zeta) | \leq \eps M,
\]
so, as $\cos \theta_0 = 1 - O_1(\theta_0^2/2)$,
\[
O_1(\delta/\pi) + (1 + O_1(\delta^2/2))(r - \zetacheck_d + O_1(3 d^3 \delta)) = O_1(\eps M),
\]
and thus, using $r \leq 4 d$, $\zetacheck_d \leq d$, and $\delta = 8 \pi d^4 \sqrt{\eps} \leq 1$,
\begin{align}
r - \zetacheck_d &= O_1(\eps M + \delta/\pi + 3 d^3 \delta + \frac{5}{2} \delta^2 d + \frac{3}{2} d^3 \delta^3) \\
&= O_1 \Bigl(4 \eps d + \delta ( \frac{1}{\pi} + 3 d^3 + \frac{5}{2} d + \frac{3}{2} d^3) \Bigr) \\
&= O_1(200 d^7 \sqrt{\eps})
\end{align}

To conclude the proof of \eqref{Seps-estimate-balls}, we use the fact that
\[
\|x - \xo\|_2 \leq  \bigl| \|x\|_2 - \|\xo\|_2   \bigr| + ( \|\xo\|_2 +  \bigl| \|x\|_2 - \|\xo\|_2   \bigr| ) \thetabar_0.
\]
This fact simply says that if a 2d point is known to have magnitude within $\Delta r$ of some $r$ and is known to be within angle $\Delta \theta$ from $0$, then its Euclidean distance to the point of polar coordinates $(r,0)$ is no more than $\Delta r + (r +  \Delta r) \Delta \theta$.

Finally, we establish that $\rho_d\to 1$ as $d \to \infty$.  Note that $\rho_{d+1} = (1 - \frac{\thetacheck_d}{\pi}) \rho_d + \frac{\sin \thetacheck_d}{\pi}$ \red{ and $\rho_0 = 0$}.  It suffices to show $\rhotilde_d \to 0$, where $\rhotilde_d := 1 - \rho_d$.  \red{The following recurrence relation holds: $\rhotilde_d = (1 - \frac{\thetacheck_{d-1}}{\pi}) \rhotilde_{d-1} + \frac{\thetacheck_{d-1} - \sin \thetacheck_{d-1}}{\pi}$, with $\rhotilde_0=1$. } Using the recurrence formula \eqref{recurrence-relation} \red{ and the fact that $\thetacheck_0=\pi$}, we get that 
\begin{align}
\rhotilde_d =  \sum_{i=1}^d \frac{\thetacheck_{\red{i}-1} - \sin \thetacheck_{\red{i}-1}}{\pi} \prod_{j=i+1}^d \bigl(1 - \frac{\thetacheck_{\red{j}-1}}{\pi} \bigr)
\end{align}
using \eqref{lower-bound-thetai}, we have that
\begin{align*}
\prod_{j=i+1}^d \Bigl(1 - \frac{\thetacheck_{\red{j}-1}}{\pi} \Bigr) &\leq \prod_{j=i+1}^d \Bigl(1 - \frac{1}{\red{j}} \Bigr)
= \exp \Bigl( -\sum_{j = i+1}^d \frac{1}{\red{j}}  \Bigr) \leq \exp \Bigl( - \int_{i+1}^{d+1} \frac{1}{\red{s}} ds  \Bigr) = \frac{i+\red{1}}{d+\red{1}}
\end{align*}
Using \eqref{upper-bound-thetai} and the fact that $\thetacheck_{\red{i}-1} - \sin \thetacheck_{\red{i}-1} \leq \thetacheck_{\red{i}-1}^3 / 6$, we have that $\rhotilde_d \leq \sum_{i=1}^d \frac{\thetacheck_{\red{i}-1}^3}{6\pi}\cdot  \frac{i+\red{1}}{d+\red{1}} \to 0$ as $d \to \infty$.

\end{proof}

\subsection{Proof of WDC for Gaussian Matrices} \label{sec:WDC-Gaussian}

In this section, we establish a bound on the probability that a Gaussian matrix satisfies the WDC, provided it corresponds to a sufficiently expansive layer of a neural network.  \begin{lemma} \label{lemma:isometry-apxtapx-uniform} 
Fix $0 < \eps < 1$.  Let $W \in \R^{n \times k}$ have i.i.d. $\calN(0, 1/n)$ entries.   If $n > c k \log k$, then with probability at least $1- 8 n e^{-\gamma k}$, $W$ satisfies the WDC with constant $\eps$. Here $c, \gamma^{-1}$ are constants that depend only polynomially on $\eps^{-1}$. 
\end{lemma}

The WDC with constant $\eps$ can be written as
\begin{align}
\| \Wpx^t\Wpy  - \Qxy \| \leq \eps  \label{WpxWpyQxy}
\end{align}
for all nonzero $x,y \in \R^k$. 
 A common way to establish concentration of a random function simultaneously over an infinite number of values of $(x,y)$, like \eqref{WpxWpyQxy}, is as follows:
\begin{itemize}
\item Show concentration of the quantity with high probability for a fixed $(x,y)$.
\item Bound the Lipschitz constant of the quantity with respect to $(x,y)$.
\item Take a union bound over a net whose size is given by the Lipschitz constant.
\end{itemize}
This argument does not apply in the present case because $\Wpx^t\Wpy$ is not continuous with respect to $(x,y)$.  To deal with this lack of continuity, we form two continuous variants that are greater and less than $\Wpx^t\Wpy$, respectively, with respect the semidefinite ordering.  
We now introduce some notation in order to state these bounds.  
Let 
\[
\hmeps(z) = \begin{cases}0 & z \leq -\eps, \\ 1 + \frac{z}{\eps} & -\eps \leq z \leq 0, \\ 1 & z\geq 0,  \end{cases} \quad \text{ and } \quad 
\heps(z) = \begin{cases}0 & z \leq 0, \\  \frac{z}{\eps} & 0 \leq z \leq \eps, \\ 1 & z\geq \eps.  \end{cases}
\]
When applied to a vector, let $\hmeps$ and $\heps$ act component-wise.
Let $\wit$ be the $i$th row of $W$.  Note that $\WpxtWpy = \sum_{i=1}^n \indwx \indwy \cdot \wi \wit$, and define 
\begin{align*}
\Gmeps(x, y) :=  \sum_{i=1}^n \hmeps(\wix) \hmeps(\wiy) \wiwit \quad  \text{ and } \quad
\Geps(x,y) :=  \sum_{i=1}^n \heps(\wix) \heps(\wiy) \wiwit.
\end{align*}

As $\hmeps(z) \geq 1_{z>0}(z)$ and $\heps(z) \leq 1_{z<0}(z)$ for all $z \in \R$, we have that for all nonzero $x,y$ that $\Geps(x,y) \matrixleq \WpxtWpy \matrixleq \Gmeps(x,y)$.  Thus, it suffices to establish a matrix upper bound on $\Gmeps$ and a matrix lower bound on $\Geps$. 

First we establish a matrix upper bound on $\Gmeps(x,y)$ uniformly over all nonzero $x,y$.  For ease of exposition, in the next two Lemmas, we will take the entries of $W$ to be i.i.d. $\calN(0,1)$.  In this case, $\E[\Wpxt \Wpy]=n \Qxy$. 

\begin{lemma} \label{lemma:Gmeps-uniform} 
Fix $0 < \eps < 1$.  Let $W \in \R^{n \times k}$ have i.i.d. $\calN(0,1)$ entries.   If $n > c k \log k$, then with probability at least $1- 4 n e^{-\gamma k}$,
\[
\forall x\neq 0,y\neq 0, \quad \Gmeps(x,y) \matrixleq  n \Qxy  + 3 \eps n I_k.
\]
Here, $c$ and $\gamma^{-1}$ are constants that depend only polynomially on $\eps^{-1}$.
\end{lemma} 
\begin{proof}
Note that the entries of $W$ are assumed to have $\mathcal{N}(0,1)$ entries, and not $\mathcal{N}(0, 1/n)$ entries like in most of this paper.
In this proof, the values of the constants $c$ and $\gamma$ may change from line to line, but they are all bounded above and below, respectively, by some $\eps$-dependent constant.  Without loss of generality, let $x,y \in \spherek$. 

First, we bound $\E[\Gmeps(x,y)]$ for fixed $x, y \in \spherek$.   Noting that $\hmeps(z)  \leq 1_{z\geq -\eps}(z) = 1_{z>0}(z) + 1_{-\eps \leq z \leq 0}(z)$, we have
\begin{align}
\E[ \Gmeps(x,y) ] &\matrixleq \E \Bigl[ \sum_{i=1}^n 1_{\wix \geq -\eps} 1_{\wiy \geq -\eps} \cdot \wi \wit  \Bigr] \notag \\
&\matrixleq \E \Bigl[\sum_{i=1}^n (1_{\wix>0} 1_{\wiy>0} + 1_{-\eps \leq \wix \leq 0} + 1_{-\eps \leq \wiy \leq 0} ) \cdot  \wi \wit \Bigr] \notag \\
&=  n\Qxy +  n \E[ 1_{-\eps\leq \wix \leq 0} \cdot \wi \wit ] +  n \E[ 1_{-\eps\leq \wiy \leq 0} \cdot \wi \wit ]. \notag
\end{align}
We now bound $\E[ 1_{-\eps\leq \wix \leq 0} \cdot \wi \wit ]$.  For deriving this bound, we may take $x = e_1$ without loss of generality.  We have for $\eps < 1$, 
\begin{align}
\E[ 1_{-\eps\leq \wix \leq 0} \cdot \wi \wit  ] &= \begin{pmatrix}\zeta_1 & 0 \\ 0 & \zeta_2 \cdot I_{n-1} \end{pmatrix}, \notag \\
0 \leq \zeta_1 &\leq \int_{-\eps}^0 z^2 \frac{1}{\sqrt{2 \pi}} e^{-z^2/2} dz \leq \frac{\eps}{2}, \notag \\
0 \leq \zeta_2 &\leq \int_{-\eps}^0  \frac{1}{\sqrt{2 \pi}} e^{-z^2/2} dz \leq \frac{\eps}{2}. \notag
\end{align}
Thus, $\E[ 1_{-\eps< \wix < 0} \cdot \wi \wit ] \matrixleq \frac{\eps}{2} I_k$ for any $x \neq 0$, resulting in
\begin{align}
\E[\Gmeps(x, y)] \matrixleq n\Qxy + \eps n \cdot I_k. \label{bound-Gmeps-expectation-uniform}
\end{align}

Second, we show concentration of $\Gmeps(x,y)$ for fixed $x,y \in \spherek$.  Let $\xii = \sqrt{\hmeps(\wix)} \sqrt{\hmeps(\wiy)} \wi$.  We have
\begin{align}
\Gmeps(x,y) - \E[ \Gmeps(x,y) ] &= \sum_{i=1}^n  \Bigl( \hmeps(\wix) \hmeps(\wiy) \wiwit - \E\bigl[ \hmeps(\wix) \hmeps(\wiy) \wiwit \bigr]   \Bigr) \\&= \sum_{i=1}^n (\xiixiit - \E \xiixiit).
\end{align}
Note that $\xii$ is sub-Gaussian for all $i$ and that the sub-Gaussian norm of $\xii$ is bounded above by an absolute constant, which we will call $K$.  By the first part of Remark 5.40 in \cite{V2012}, there exist constants $c_K$ and $\gamma_K$ such that for all $t \geq 0$, with probability at least $1 - 2 e ^{-\gamma_K t^2}$,
\[
\| \Gmeps(x,y) - \E \Gmeps(x,y) \| \leq \max(\delta, \delta^2) n, \quad \text{ where } \delta = c_K \sqrt{\frac{k}{n}} + \frac{t}{\sqrt{n}}.
\]

If $n > (2/\eps)^2 c_K^2 k$, $t = \eps \sqrt{n}/2$, and $\eps < 1$, we have that with probability at least $1-2e^{-\gamma_K \eps^2 n/4}$,
\begin{align}
\| \Gmeps(x,y) - \E \Gmeps(x,y) \| \leq \eps n.   \label{bound-Gmeps-concentration-uniform} 
\end{align}

Third, we bound the Lipschitz constant of $\Gmeps$.  For $\xtilde, \ytilde \in \R^k$ we have
\begin{align}
\Gmeps(x, y) - \Gmeps(\xtilde, \ytilde) &= \sum_{i=1}^n \bigl[ \hmeps(\wix) \hmeps(\wiy) - \hmeps(\wixtilde) \hmeps(\wiytilde) \bigr] \wiwit \notag \\
&= \sum_{i=1}^n \bigl[ \hmeps(\wix) \bigl(\hmeps(\wiy) - \hmeps(\wiytilde) \bigr)  \notag \\[-.75em]
&\quad \quad\quad + \hmeps(\wiytilde) \bigl(\hmeps(\wix) - \hmeps(\wixtilde) \bigr)  \bigr] \wiwit \notag \\
&= \Wt \bigl[ \diag \bigl( \hmeps(Wx) \bigr) \diag \bigl(\hmeps(Wy) - \hmeps(W\ytilde) \bigr) \notag\\
&\quad \quad \quad+  \diag \bigl( \hmeps(W \ytilde) \bigr) \diag \bigl(\hmeps(Wx) - \hmeps(W\xtilde) \bigr)  \bigr] W. \notag
\end{align}
Thus,
\begin{align*}
\| \Gmeps(x,y) - \Gmeps(\xtilde, \ytilde) \| 
&\leq \| W\|^2  \Bigl[  \| \hmeps(W x)\|_\infty  \|  \hmeps(Wy) - \hmeps(W\ytilde) \|_\infty \notag \\
&\quad \quad \quad +  \| \hmeps(W \ytilde)\|_\infty  \|  \hmeps(Wx) - \hmeps(W\xtilde) \|_\infty  \Bigr ] \\
&\leq  \| W\|^2 \Bigl[ \max_{i \in [n]} \big | \hmeps(\wiy) - \hmeps(\wiytilde) \big|  \notag \\
& \quad \quad \quad  +\max_{i \in [n]}\big | \hmeps(\wix) - \hmeps(\wixtilde) \big|  \Bigr ] \\
&\leq \|W\|^2 \Bigl[ \max_{i \in [n]} \frac{1}{\eps} \big| \wi \cdot(x-\xtilde) \big| + \max_{i \in [n]} \frac{1}{\eps} \big| \wi \cdot(y-\ytilde) \big|  \Bigr] \\
&\leq \|W\|^2 \frac{1}{\eps} \Bigl(\max_{i \in [n]} \| \wi\|_2 \Bigr) ( \| x-\xtilde \|_2  +  \| y-\ytilde \|_2 ) 
\end{align*}
where the second inequality follows because $|\hmeps(z)| \leq 1$ for all $z$, and the third inequality follows because $\hmeps$ is $1/\eps$-Lipschitz.

Let $E_1$ be the event that $\|W\| \leq 3 \sqrt{n}$.  By Corollary 5.35 in \cite{V2012}, we have that $\PP(E_1) \geq 1 - 2e^{-n/2}$, if $n \geq k$.  Let $E_2$ be the event that $\max_{i \in [n]} \|\wi \|_2 \leq 2 \sqrt{k}$.    By a single-tailed variant of Corollary 5.17 in \cite{V2012}, there exists a constant $\gamma_0$ such that for fixed $i$,  $\| \wi\|_2 \leq 2 \sqrt{k}$ with probability at least $1 - e^{-\gamma_0 k}$.  Thus, $\PP(E_2) \geq 1- n e^{-\gamma_0 k}$. 

On $E_1 \cap E_2$, we have  
\begin{align}
\| \Gmeps(x,y) - \Gmeps(\xtilde, \ytilde) \| \leq \frac{18 n \sqrt{k}}{\eps} (\|x-\xtilde\|_2 + \|y-\ytilde\|_2) .   \label{lipschitz-constant-Gmeps-uniform}
\end{align}
for all $x,y,\xtilde, \ytilde \in \spherek$. 

Finally, we complete the proof by a covering argument.  Let $\Ndelta$ be a $\delta$-net on $\spherek$ such that $| \Ndelta| \leq (3/\delta)^k$.  Take $\delta = \frac{\eps^2}{36 \sqrt{k}}$.   Combining \eqref{bound-Gmeps-concentration-uniform} and \eqref{bound-Gmeps-expectation-uniform}, we have 
\[
\forall x,y \in \Ndelta, \quad \Gmeps(x,y) \matrixleq \E \Gmeps(x,y) + \eps  n I_k \matrixleq n\Qxy + 2 \eps n I_k.
\]
with probability at least $1 - 2 | \Ndelta| e^{-\gamma_K \eps^2 n/4} \geq 1 - 2 (\frac{3}{\delta})^k e^{- \gamma_K \eps^2 n/4 } \geq 1 - 2 e^{-\gamma_K \eps^2 n/4 + k \log (3\cdot 36 \sqrt{k}/\eps^2)}$.  If $n > \tilde{c} k \log k$, for some $\tilde{c} = \Omega (\eps^{-2} \log \eps^{-1})$, then this probability is at least $1 - 2 e^{-\tilde{\gamma} n}$ for some $\tilde{\gamma}=O(\eps^2)$. 
For any $x,y \in \spherek$, let $\xtilde, \ytilde \in \Ndelta$ be such that $\|x-\xtilde\|_2 < \delta$ and $\|y-\ytilde\|_2 < \delta$.  By \eqref{lipschitz-constant-Gmeps-uniform}, we have that 
\[
\forall x\neq0, y\neq 0, \quad \Gmeps(x,y) \matrixleq \Gmeps(\xtilde, \ytilde) + \frac{18 n\sqrt{k}}{\eps} 2 \delta I_k \matrixleq n \Qxy + 3 \eps n I_k.
\]
In conclusion, the result of this lemma holds if $n > (2/\eps)^2 c_K^2 k$ and $n > \tilde{c} k \log k$, with probability at least $1 - 2 e^{-\gamma_K \eps^2 n/4} - 2 e^{-n/2} - n e^{-\gamma_0 k} - 2 e^{-\tilde{\gamma} n} > 1 - 4 n e^{-\gamma k}$ for some $\gamma = O(\eps^2)$ and $\tilde{c} = \Omega (\eps^{-2} \log \eps^{-1})$.
\end{proof}

\begin{lemma} \label{lemma:Geps-uniform} 
Fix $0 < \eps < 1$.  Let $W \in \R^{n \times k}$ have i.i.d. $\calN(0,1)$ entries.   If $n > c k \log k$, then with probability at least $1- 4 n e^{-\gamma k}$,
\[
\forall x\neq 0, y \neq 0, \quad \Geps(x,y) \matrixgeq n \Qxy - 3 \eps n I_k.
\]
Here, $c$ and $\gamma^{-1}$ are constants that depend only polynomially on $\eps^{-1}$.
\end{lemma} 
\begin{proof}
The proof follows that of Lemma \ref{lemma:Gmeps-uniform} exactly. 

First, we bound $\E[\Geps(x,y)]$ for fixed $x,y \in \spherek$.  For deriving this bound, we take $x = e_1$ without loss of generality. Noting that $\heps(z) \geq 1_{z>0}(z) -  1_{0\leq z \leq \eps}(z)$ for all $z$, we have
\begin{align*}
\E[ \Geps(x,y) ] &\matrixgeq   n\Qxy -   n\E[ 1_{0\leq \wix \leq \eps} \cdot \wi \wit ] -   n \E[ 1_{0\leq \wiy \leq \eps} \cdot \wi \wit ].
\end{align*}

 We have  
\begin{align*}
\E[ 1_{0\leq \wix \leq \eps} \cdot \wi \wit  ] &= \begin{pmatrix}\zeta_1 & 0 \\ 0 & \zeta_2 \cdot I_{n-1} \end{pmatrix}, \\
0 \leq \zeta_1 &\leq \int_{0}^\eps z^2 \frac{1}{\sqrt{2 \pi}} e^{-z^2/2} dz \leq \frac{\eps}{2},\\
0 \leq \zeta_2 &\leq \int_{0}^\eps  \frac{1}{\sqrt{2 \pi}} e^{-z^2/2} dz \leq \frac{\eps}{2}.
\end{align*}
Thus, $\E[ 1_{0\leq \wix \leq \eps} \cdot \wi \wit ] \matrixleq \frac{\eps}{2} I_k$ for any $x \neq 0$, resulting in
\begin{align*}
\E[\Geps(x,y)] \matrixgeq n \Qxy - \eps n \cdot I_k. 
\end{align*}

Second, the same argument as in Lemma \ref{lemma:Gmeps-uniform} provides that for fixed $x,y\in \spherek$, 
if $n > (2/\eps)^2 c_K^2 k$, then we have that with probability at least $1-2e^{-\gamma_K \eps^2 n/4}$, 
\begin{align*}
\| \Geps(x,y) - \E \Geps(x,y) \| \leq \eps n.   
\end{align*}

Third, the same argument as in Lemma \ref{lemma:Gmeps-uniform} provides that on the event $E_1\cap E_2$, we have
\begin{align*}
\| \Geps(x,y) - \Geps(\xtilde, \ytilde) \| \leq \frac{18 n\sqrt{k}}{\eps} \bigl[ \|x-\xtilde \|_2 + \|y - \ytilde\|_2 \bigr].   
\end{align*}
for all $x,y,\xtilde, \ytilde \in \spherek$.

Finally, we complete the proof by an identical covering argument as in Lemma \ref{lemma:Gmeps-uniform}.  We have that if $n > c_0 k \log k$ then with probability at least $1 - 4 n e^{-\gamma k}$, 
\[
\forall x, y \in \spherek, \Geps(x,y) \matrixgeq n \Qxy   - 3 \eps n I_k.
\]
\end{proof}

We may now prove Lemma \ref{lemma:isometry-apxtapx-uniform}.
\begin{proof}[Proof of Lemma \ref{lemma:isometry-apxtapx-uniform}]

It suffices to show
\[
\forall x,y \in \R^k, \quad \|\WpxtWpy -  n\Qxy \| \leq  3 n\eps,
\]
where $A$ has i.i.d. $\mathcal{N}(0,1)$ entries.
The result is immediate if $x=0$ or $y = 0$.
  For all nonzero $x$ and $y$ , we have $\Geps(x,y) \matrixleq \WpxtWpy \matrixleq \Gmeps(x,y)$.  The lemma then follows directly from Lemmas \ref{lemma:Gmeps-uniform} and \ref{lemma:Geps-uniform}.
\end{proof}

\subsection{Proof of RRIC} \label{sec:RRIC-Gaussian}

We will make use of a standard concentration result used in proving the Restricted Isometry Property from compressed sensing \cite{Baraniuk2008}.
\begin{lemma}[Variant of Lemma 5.1 in \cite{Baraniuk2008}] \label{lemma:rip-baraniuk} 
Let $A \in \R^{m \times n}$ have i.i.d. $\calN(0, 1/m)$ entries.  Fix $0< \eps <1, k < m$.  Fix a subspace $T \subset \R^n$ of dimension $k$.  With probability at least $1 - (c_1/\eps)^k e^{-\gamma_1 \eps m}$, 
\[
(1-\eps) \|x\|_2^2 \leq \|Ax\|_2^2 \leq (1+\eps) \|x\|_2^2, \quad \forall x \in T,
\]
and
\[
| \langle Ax, Ay \rangle - \langle x, y \rangle | \leq \eps \|x\|_2 \|y\|_2, \quad \forall x, y \in T.
\]
Let $V = \bigcup_{i=1}^M V_i$ and $W = \bigcup_{j=1}^N W_j$, where $V_i$ and $W_j$ are subspaces of $\R^n$ of dimension at most $k$ for all $i,j$.  Then,
\[
| \langle Ax, Ay \rangle - \langle x, y \rangle| \leq \eps \|x\|_2 \|y\|_2, \quad \forall x \in V, y \in W,
\]
with probability at least $ 1 - MN(c_1/\eps)^{2k} e^{-\gamma_1 \eps m}. $
Here, $c_1$ and $\gamma_1$ are universal constants.
\end{lemma}
\begin{proof}
This proof is an immediate extension of Lemma 5.1 in \cite{Baraniuk2008}.  As $A$ is Gaussian, we may take $T$ to be the span of $k$ standard basis vectors and directly apply the lemma in that paper. The inner product form follows by a standard argument based on the parallelogram identity.  The last inequality holds by applying the second inequality to all subspaces of the form $\Span(V_i, W_j)$, which have dimension at most $2k$, and by applying a union bound.
\end{proof}

In order to apply Lemma \ref{lemma:rip-baraniuk}, we now provide an upper bound for the number of subspaces that arise from the objective $f$.

\begin{lemma} \label{lemma:combinatorial-bound-Apluszero} 
Let $V$ be a subspace of $\R^k$.  Let $W \in \R^{n \times k}$ have i.i.d. $\calN(0, 1/n)$ entries.    With probability 1, 
\begin{align}
| \{ \diag(Wv > 0) W \mid v \in V\} | &\leq 10 n^{\dim V}.
\end{align}
\end{lemma}
\begin{proof}
Let $\ell = \dim V$.  By rotational invariance of Gaussians, we may take $V= \Span(e_1, \ldots, e_\ell)$ without loss of generality.  Without loss of generality, we may let $W$  have dimensions $n \times \ell$ and take $V = \R^\ell$. 

We will appeal to a classical result from sphere covering \cite{wendel1962problem}.  If $n$ hyperplanes in $\R^\ell$ contain the origin and are such that the normal vectors to any subset of $\ell$ of those hyperplanes are independent, then the complement of the union of these hyperplanes is partitioned into at most 
\[
2 \sum_{i = 0}^{\ell-1}{n-1 \choose i}
\]
disjoint regions. Note that for fixed $W$,  $|\{ \diag(Wv>0) W \mid v \in \R^\ell\}|$ equals the number of binary vectors of the form $(1_{\wi \cdot v > 0})_{i \in [n]}$ for $ v \in \spherel$.  Each such binary vector corresponds uniquely to one of the disjoint regions given by partitioning the unit sphere in $\R^\ell$ by the $n$ half-spaces going through the origin with normal vectors $\{\wi\}_{i \in [n]}$.
With probability 1, any subset of $\ell$ rows of $W$ are linearly independent, and thus,
\[
|\{ \diag(Wv>0) W \mid v \in \R^\ell\}| = 2 \sum_{i = 0}^{\ell-1}{n-1 \choose i} \leq 2 \ell \Bigl(\frac{en}{\ell}\Bigr)^\ell \leq 10 n^\ell.
\]
where the first inequality uses the fact that ${n \choose \ell} \leq (en/\ell)^\ell$ and the second inequality uses that $2 \ell (e/\ell)^\ell \leq 10$ for all $\ell \geq 1$. 

%
%
\end{proof}

\begin{lemma} \label{lemma:num-apx-aox-bpx-box} 
Let $W_i \in \R^{n_i \times n_{i-1}}$ have i.i.d. $\calN(0, 1/n_i)$ entries for $i = 1,\ldots,d$.  Let $k = n_0$.   Then, with probability 1,
\begin{align}
| \{ \Wipx \mid x \neq 0\} | \leq 10^i n_1^k n_2^k \cdots n_i^k. 
\end{align}

\end{lemma}
\begin{proof}
Recall that $\Wipxn{1} = \diag(W_1 x>0) W_1$ and $\Wipx = \diag(W_i \Wipxn{i-1} \cdots \Wipxn{1} x>0) W_i$. 
The case of $i=1$ holds with probability 1 by applying Lemma \ref{lemma:combinatorial-bound-Apluszero} with $V = \R^k$. 

Next, we establish the $i=2$ case. Let $\calWp = \{ \Wpx \mid x \neq 0\}$.  Note that 
\[
\forall x \neq 0, \quad \Wipxn{2} \in \bigcup_{\What \in \calWp} \{ \diag( W_2 v>0) W_2 \mid v \in \range \What \}.
\]
Let the random variable $\XWhatWtwo = | \{ \diag(W_2 v>0)W_2 \mid v \in \range \What \} |$.  Note that by Lemma \ref{lemma:combinatorial-bound-Apluszero}, for any fixed $\What$, $\PP_{W_2}(\XWhatWtwo \leq  10 n_2^k) = 1$.  Hence, conditioned on the probability 1 event $E:=\{|\calWp|\leq 10 n_1^k\}$, we have  $\PP_{W_2}(\sum_{\What \in \calWp} \XWhatWtwo \leq 10^2 n_1^k n_2^k) = 1$. 
Thus, 
\begin{align*}
\PP_{W_1,W_2} \Biggl( \sum_{\What \in \calWp} \XWhatWtwo \leq 10^2 n_1^k n_2^k \Biggr) 
& = \int_{\Omega_{1}} \PP_{W_2}\Biggl( \sum_{\What \in \calWp} \XWhatWtwo \leq 10^2 n_1^k n_2^k \Biggr) d\mu_{1}\\
& = \int_{E} \PP_{W_2}\Biggl( \sum_{\What \in \calWp} \XWhatWtwo \leq 10^2 n_1^k n_2^k \Biggr) d\mu_{1} \\
 &= 1,
\end{align*}
where $\Omega_{1}$ and $\mu_{1}$ are the state space and probability measure for the random variable $W_1$.  Hence, $|\{\Wipxn{2} \mid x \neq 0\}| \leq 10^2 n_1^k n_2^k$ with probability 1.

The case of larger $i$ follows by repeating the logic above.


\end{proof}

We can now show concentration of $\vxxo$ to its expectation with respect to $C$.

\begin{lemma}\label{lemma:ABCCBA-ABBA-xy} 
Fix $0 < \eps < 1$.  
%
Let $W_i \in \R^{n_i \times n_{i-1}}$, have i.i.d. $\calN(0, 1/n_i)$ entries for $i = 1,\ldots,d$.  Let $A \in \R^{m \times n_d}$ have i.i.d. $\calN(0, 1/m)$ entries that are independent from all $W_i$.  If $m > c d k \log(n_1 n_2 \cdots n_d)$, then with probability at least $1 -  e^{-\gamma m}$, 
\begin{align}
\forall x, y \in \R^k, \quad \| (\PiWd)^t \AtA (\PiWdy) -(\PiWd)^t (\PiWdy)  \| \leq \eps \Pioned \| \Wipx\| \|\Wipy\| . \label{eqn:ABCCBAxy-Wcase}
\end{align}

Here, $c$ and $\gamma^{-1}$ are constants that depend only polynomially on $\eps^{-1}$.
\end{lemma}
\begin{proof}
For pedagogical purposes, we first establish the lemma in the $d=2$ case.  It suffices to show that $\forall x, y,w, v \in \spherek$,
\begin{align*}
 | \langle A \Wipxn{2} \Wipxn{1} w, A \Wipyn{2} \Wipyn{1} v \rangle &- \langle \Wipxn{2} \Wipxn{1}  w , \Wipyn{2} \Wipyn{1}  v\rangle | \\&\leq \eps \|\Wipxn{1}\| \| \Wipxn{2} \| \|\Wipyn{1}\| \| \Wipyn{2}\|.
\end{align*}

In order to apply Lemma \ref{lemma:rip-baraniuk}, we will show that $\{ \Wipxn{2} \Wipxn{1} w \mid x, w \in \spherek \}$ is a subset of a union of at most $10^3 (n_1^2 n_2)^k$ subspaces of dimension at most $k$. 
For fixed \red{$W_1,W_2$}, let $\calAp = \{\Wipxn{1} \mid x \neq 0\}$ and $\calBp = \{\Wipxn{2} \mid x \neq 0\}$.  
By Lemma \ref{lemma:num-apx-aox-bpx-box}, there exists a probability 1 event, E, over \red{$(W_1,W_2)$} on which $|\calAp| \leq 10 n_1^k$ and $|\calBp| \leq 10^2 n_1^k n_2^k$. 
On $E$, 
\[
|\{ \Wipxn{2} \Wipxn{1} \mid x \neq 0\}| \leq 10^3 (n_1^2 n_2)^k.
\]
 Note that $\dim \range(\Wipxn{2} \Wipxn{1}) \leq k$ for all $x \neq 0$.  Hence $\{ \Wipxn{2} \Wipxn{1} w \mid x, w \in \spherek \} \subset V$, where $V$ is a union of at most $10^3 (n_1^2 n_2)^k$ subspaces of dimensionality at most $k$. 

By applying \red{the second half of} Lemma \ref{lemma:rip-baraniuk} to the sets $V$ and $V$, we get that for fixed \red{$W_1,W_2$}, 
\begin{align}
| \langle A \Wipxn{2} \Wipxn{1} w, A \Wipyn{2} \Wipyn{1} v \rangle &- \langle \Wipxn{2} \Wipxn{1} w, \Wipyn{2} \Wipyn{1} v \rangle |  \notag
\\&\leq \eps \|\Wipxn{2} \Wipxn{1} w\|_2 \| \Wipyn{2} \Wipyn{1} v \|_2, \quad \forall x,y, w, v \in \spherek \label{CawCavAwAv}
\end{align}
with probability at least $1 - 10^3 (c_1 n_1^2 n_2/\eps)^{2k} e^{-\gamma_1 \eps m} \geq 1 - e^{-\gamma_2 m}$, provided $m \geq \ctilde k \log(n_1 n_2)$, for universal constants $c_1, \gamma_1$ and for some $\gamma_2 = \frac{\gamma_1 \eps}{2}, \ctilde =\Omega(\eps^{-1} \log \eps^{-1})$. 

Integrating over the probability space of \red{$(W_1,W_2)$, independence of $A$ and $(W_1,W_2)$} implies that  \eqref{CawCavAwAv} holds for random \red{$(W_1,W_2)$} with the same probability bound.
Continuing from \eqref{CawCavAwAv}, we have
\begin{align*}
| \langle A \Wipxn{2} \Wipxn{1} w, A \Wipyn{2} \Wipyn{1} v \rangle &- \langle \Wipxn{2} \Wipxn{1} w, \Wipyn{2} \Wipyn{1} v \rangle | \\&\leq \eps \|\Wipxn{1}\|_2 \|\Wipyn{1}\|_2 \|\Wipxn{2}\|_2 \|\Wipyn{2}\|_2 
\end{align*}
$\forall x,y, w, v \in \spherek$ with probability at least $1 -  e^{-\gamma m}$. for some $\gamma>0$. 

For the case of $d \geq 2$, the lemma follows similarly.  We have 
\[
|\{ \PiWd x \mid x \neq 0\}| \leq 10^{(d^2)} (n_1^d n_2^{d-1} \cdots n_{d-1}^2 n_d)^k
\] on the probability 1 event.  The analogous bound to \eqref{CawCavAwAv} holds with probability at least $1 - 10^{(d^2)}(c_1 n_1^d n_2^{d-1} \cdots n_{d-1}^2 n_d /\eps)^{2k} e^{-\gamma_1 \eps m} \geq 1 - e^{-\gamma_2 m}$, provided $m \geq \ctilde d k \log(n_1 n_2 \cdots n_d)$, for some $\gamma_2 = \frac{\gamma_1 \eps}{2}, \ctilde = \Omega(\eps^{-1} \log \red{\eps^{-1}})$.

\end{proof}

\subsection{Proof of Theorem \ref{thm-multi-layer}}

Theorem \ref{thm-multi-layer} can be proved by combining Lemmas \ref{lemma:isometry-apxtapx-uniform} and \ref{lemma:ABCCBA-ABBA-xy}.

\subsubsection*{Acknowledgment}
PH is partially supported by NSF Grant DMS-1464525.



\bibliographystyle{plain}
\bibliography{refs}

  \appendix

\end{document}